\documentclass[11pt]{article}

\pdfoutput=1

\usepackage[margin=1in]{geometry}

\usepackage[
  bookmarks=true,
  bookmarksnumbered=true,
  bookmarksopen=true,
  pdfborder={0 0 0},
  breaklinks=true,
  colorlinks=true,
  linkcolor=black,
  citecolor=black,
  filecolor=black,
  urlcolor=black,
]{hyperref}
\usepackage[round]{natbib}
\usepackage{amsmath,amsthm,amssymb}

\newcommand{\doublefigxscale}{0.85}
\newcommand{\allfigscale}{1}
\newcommand{\allfigfont}{\normalsize}
\newcommand{\doublefigfont}{\normalsize}

\usepackage[capitalise]{cleveref}
\usepackage{nicefrac}
\usepackage{subfigure}
\usepackage{isomath}
\usepackage{mathtools}

\usepackage{fnbreak}
\usepackage{tikz}
\usepackage{pgfplots}
\pgfplotsset{compat=1.12}
\usetikzlibrary{calc}
\usepackage{graphicx, color}
\usetikzlibrary{intersections,plotmarks}

\allowdisplaybreaks

\theoremstyle{plain}
\newtheorem{theorem}{Theorem}
\newtheorem{corollary}{Corollary}
\newtheorem{lemma}{Lemma}
\newtheorem{proposition}{Proposition}
\newtheorem{open-problem}{Open Problem}
\newtheorem{sublemma}[theorem]{Sublemma}
\Crefname{sublemma}{Sublemma}{Sublemmas}

\Crefname{equation}{Equation}{Equations}
\Crefname{figure}{Figure}{Figures}

\theoremstyle{definition}
\newtheorem{definition}{Definition}

\newcommand{\crefpart}[2]{\cref{#1}(\labelcref{#1-#2})}
\newcommand{\RR}{\mathbb{R}_{+}}
\newcommand{\eqdef}{\triangleq}
\newcommand{\prob}[2]{\Pr_{#1}\left[#2\right]}
\newcommand{\expect}[2]{\mathbb{E}_{#1}\left[#2\right]}
\newcommand{\qo}{q_*}
\DeclareMathOperator{\OPT}{OPT}
\newcommand{\ERM}{\mathit{ERM}}

\hypersetup{
  pdfauthor      = {Moshe Babaioff <moshe@microsoft.com>, Yannai A. Gonczarowski <yannai@gonch.name>, Yishay Mansour <mansour@tau.ac.il>, and Shay Moran <shaymoran1@gmail.com>},
  pdftitle       = {Are Two (Samples) Really Better Than One? On the Non-Asymptotic Performance of Empirical Revenue Maximization},
}

\title{Are Two (Samples) Really Better Than One?\texorpdfstring{\\}{ }On the Non-Asymptotic Performance of\texorpdfstring{\\}{ }Empirical Revenue Maximization}
\author{Moshe Babaioff\thanks{Microsoft Research, \emph{E-mail}: \href{mailto:moshe@microsoft.com}{moshe@microsoft.com}.} \and Yannai A. Gonczarowski\thanks{Einstein Institute of Mathematics, Rachel \& Selim Benin School of Computer Science \& Engineering, and Federmann Center for the Study of Rationality, The Hebrew University of Jerusalem, Israel; and Microsoft Research. \emph{E-mail}: \href{mailto:yannai@gonch.name}{yannai@gonch.name}.} \and Yishay Mansour\thanks{Tel Aviv University, Israel, and Google Research, Israel. \emph{E-mail}: \href{mailto:mansour@tau.ac.il}{mansour@tau.ac.il}.
The research was done while author was co-affiliated with Microsoft Research.} \and Shay Moran\thanks{Institute for Advanced Study, Princeton. \emph{E-mail}: \href{mailto:shaymoran1@gmail.com}{shaymoran1@gmail.com}. The research was done while author was co-affiliated with Microsoft Research.}}
\date{July 2, 2018}

\begin{document}
\maketitle

\begin{abstract}
The literature on ``mechanism design from samples,'' which has flourished in recent years at the interface of economics and computer science, offers a bridge between the classic computer-science approach of worst-case analysis (corresponding to ``no samples'') and the classic economic approach of average-case analysis for a given Bayesian prior (conceptually corresponding to the number of samples tending to infinity). Nonetheless, the two directions studied so far are two extreme and almost diametrically opposed directions: that of asymptotic results where the number of samples grows large, and that where only a single sample is available. In this paper, we take a first step toward understanding the middle ground that bridges these two approaches: that of a fixed number of samples greater than one. In a variety of contexts, we ask what is possibly the most fundamental question in this direction: \emph{are two samples really better than one sample?}. We present a few surprising negative results, and complement them with our main result: showing that the worst-case, over all regular distributions, expected-revenue guarantee of the \emph{Empirical Revenue Maximization} algorithm given two samples is greater than that of this algorithm given one sample. The proof is technically challenging, and provides the first result that shows that \emph{some} deterministic mechanism constructed using two samples can guarantee more than one half of the optimal revenue.
\end{abstract}

\section{Introduction}

Arguably the simplest \emph{revenue maximization} problem is that of maximizing the revenue of a single buyer from selling a single item to a single bidder. In this problem, the seller customarily possesses some prior information about the buyer, traditionally modeled via a distribution from which the valuation (maximum willingness to pay for the item) of the buyer is drawn, and the seller's task in this \emph{Bayesian} model is to devise a selling mechanism that maximizes her \emph{expected} revenue over this distribution. This classic problem was completely resolved in the seminal paper of \citet{Myerson81}, who showed that the optimal mechanism (among all truthful, possibly even randomized, mechanisms) is to offer the item for a take-it-or-leave-it price tailored for the given prior distribution.

In recent years, the literature at the interface of economics and computer science, influenced by the newly found popularity of machine learning, has seen the rise of a line of work that relaxes the assumption of complete knowledge of the underlying distribution by the seller, to the assumption of the seller having access to samples from this distribution. In a sense, this model offers a bridge, via the number of samples that are available to the seller, between the classic computer-science approach of worst-case analysis (corresponding to ``no samples'') and the above-mentioned classic economic approach of average-case analysis for a given prior distribution (conceptually corresponding to the number of samples tending to infinity). Nonetheless, all of the results that we know of in this vein are in one of two extreme and almost diametrically opposed directions: one looking at asymptotic results where the number of samples grows large  \citep{CR14,MR15,MR16,DHP16,RS16,GN17,BSV16,ABGMMY17,BSV17}, and the other asking what can be done with a single sample \citep{DRY10,HMR15,FILS15}. For example, a result of the former direction would tell us that under certain conditions, a number of samples that is polynomial in certain parameters of the problem suffices for attaining a certain approximation to the optimal revenue with high probability, while a result of the latter direction would tell us that under certain conditions, access to a single sample from the buyer's distribution allows the seller to design a mechanism that attains some constant fraction of the optimal revenue in expectation. In this paper, we take a first step towards understanding the middle ground that bridges these two approaches: that of a fixed number of samples greater than one. In particular, we ask what is possibly the most fundamental question in this direction: \emph{are two samples really better than one sample?}

To understand the specific context in which we ask the above question, and why it is more involved than may be expected, we zoom-in to provide some more context. Arguably the most natural algorithm for pricing a good given samples from an underlying distribution is the \emph{Empirical Revenue Maximization (henceforth ERM)} algorithm, which sets the price to be the one that would maximize the expected revenue over the \emph{empirical distribution} --- the uniform distribution over the given samples. For a single sample, this means offering the good for a price that equals this sample, which is shown by \citet{DRY10} to guarantee a revenue of one half (!) of the optimal revenue when the underlying distribution is regular\footnote{Regularity (sometimes called Myerson-regularity) is a standard mild restriction on valuation distributions. It is well known that without any restriction on the possible class of valuation distributions, no revenue guarantees can be proven when constructing a mechanism from a finite number of samples. To see this, consider, for small $\varepsilon$, a distribution that attains the value $0$ with probability~${1-\varepsilon}$ and the value $\nicefrac{1}{\varepsilon}$ with probability $\varepsilon$. The optimal expected revenue of $1$ is attained by posting a price of $\nicefrac{1}{\varepsilon}$, however it is futile to try and learn even the order-of-magnitude of this price (for arbitrarily small $\varepsilon$) from a finite number of samples that does not itself depend on $\varepsilon$.}.\footnote{To be completely clear, the guarantee here is that for any regular distribution, the expected revenue from this mechanism, where the expectation is taken both over the given sample and over the bidder's valuation (which are drawn i.i.d.\ from the underlying regular distribution), is at least half of the optimal expected revenue, where this expectation is taken over the bidder's valuation.}
\citet{HMR15} have in fact shown that this guarantee of one half cannot be improved upon by any deterministic mechanism that is designed based on a single sample from a regular distribution.\footnote{\citet{FILS15} have nonetheless been able to slightly improve upon this guarantee of one half using a randomized pricing mechanism.} As the number of samples grows large, \citet{HMR15} show that ERM attains revenue that asymptotically tends to optimal for regular distributions. The main question that we ask in this paper is whether, for regular distributions, the worst-case guarantee of ERM constructed based on two samples is better, the same, or worse, then the one-half worst-case guarantee of ERM constructed based on one sample.

While it is clear that the optimal method to price an item based on two samples guarantees at least as much revenue as the optimal method to price an item based on a single sample --- indeed, one could always ignore the second sample and only use the first sample as in the single-sample case --- it is far less clear that ERM when based on two samples should have a better worst-case guarantee than ERM when based on a single sample.\footnote{For example, with many samples, ERM is in fact known to \emph{not} be a worst-case-optimal pricing algorithm, and to be inferior in this sense to \emph{guarded ERM} \citep{DRY10,CR14}.} To drive this point home, we present two seemingly-similar problems, for which we show that increasing the number of samples has an undesired effect on the revenue of ERM, and an additional problem where another natural notion of monotonicity fails to hold for ERM. To phrase these three problems, it will be convenient to use the following notation: given a distribution $F$ and a natural number $n$, we denote by $\ERM(F,n)$ the expected revenue over $F$ when pricing an item according to ERM given $n$ samples from the (underlying) distribution $F$.

\paragraph{Problem 1: Is it true that for any fixed underlying regular distribution $F$ and any number $n$, it holds that $\ERM(F,n+1)\ge\ERM(F,n)$?} The answer, even for $n=1$, turns out to be \textbf{No!} In fact, it turns out that for certain distributions $F$, taking more samples ``confuses'' ERM, hurting revenue: $\ERM(F,2)<\ERM(F,1)$.

\begin{sloppypar}
\paragraph{Problem 2: Is it true that for any two fixed regular distributions $F$ and $G$, if $\ERM(F,n)>\ERM(G,n)$ then also $\ERM(F,n+1)>\ERM(G,n+1)$?} The answer, even for $n=1$, turns out to be \textbf{No!} In fact, it turns out that for certain distributions $F$ and $G$, while $\ERM(F,1)>\ERM(G,1)$, it is in fact the case that $\ERM(F,2)<\ERM(G,2)$.
\end{sloppypar}

\paragraph{Problem 3: Is it true that for any two fixed regular distributions $F$ and $G$ such that $F$ first-order stochastically dominates $G$, and for any number $n$, it holds that ${\ERM(F,n)\ge\ERM(G,n)}$?} While for $n=1$ this is known to hold, the answer, already for $n=2$, turns out to be \textbf{No!} In fact, it turns out that for certain such distributions $F$ and $G$, while the revenue from each fixed posted price is higher from $F$ than it is from $G$, the structure of $F$ ``confuses'' ERM when based on two samples, hurting revenue by causing ERM to post lower-revenue~prices.

\vspace{.5em}
The analyses of the above three problems are given in \cref{negative}.
Despite the surprising negative answers surveyed above to all three problems, we do manage to show monotonicity in the sense that the \emph{worst-case guarantee} of the price computed by ERM based on two samples from a regular distribution is strictly higher than one half of the optimal expected revenue obtained by setting a posted price tailored specifically to the underlying distribution.
To formalize this result, which is our main result, it will be convenient to use the following notation: given a distribution~$F$, we denote by $OPT(F)$ the highest expected revenue over~$F$ attained by the optimal truthful mechanism (which, recall, is a posted-price mechanism).

\begin{theorem}[See \cref{more-than-half,monotonicity}]\label{intro-theorem}
There exists $c>\nicefrac{1}{2}$ such that for every regular distribution $F$, we have $\ERM(F,2)\!>\!c\cdot OPT(F)$.
In particular, $\inf_{\text{regular $F$}}\frac{\ERM(F,2)}{OPT(F)}>\frac{1}{2}=\inf_{\text{regular $F$}}\frac{\ERM(F,1)}{OPT(F)}$.
\end{theorem}

We note that to the best of our knowledge, \cref{intro-theorem} is in fact the first result to show that \emph{some} deterministic mechanism constructed using two samples can guarantee more than one half of the optimal revenue for every regular distribution.\footnote{We emphasize again, as mentioned above, that \citet{FILS15} have constructed a \emph{randomized} mechanism for \emph{one sample} that guarantees more than one half of the optimal revenue in expectation. Our mechanism is thus the first \emph{deterministic} mechanism constructed using two samples that guarantees more than one half of the optimal revenue. (Incidentally, our lower bound beats one half by orders of magnitude more than the lower bound of \citet{FILS15} --- see \cref{main}.)} Proving \cref{intro-theorem} turns out to be a considerably more technically challenging than may have been expected (or rather, considerably more technically challenging than may have been expected before observing the negative answers to the above three problems, which may be seen as evidence that the proof of \cref{intro-theorem} should be challenging), up to the point that extending our methods even for comparing ERM for two and for three samples (let alone for higher values of $n$) seems intractable. The main problem that we leave open, therefore, is whether the monotonicity of $\ERM$ that is uncovered in \cref{intro-theorem} when going from a single sample to two samples, holds for any number of samples.

\begin{open-problem}\label[open-problem]{all-n}
Is it true that $\inf_{\text{regular $F$}}\frac{\ERM(F,n+1)}{OPT(F)}>\inf_{\text{regular $F$}}\frac{\ERM(F,n)}{OPT(F)}$ for every $n$?
\end{open-problem}

The rest of this paper is structured as follows. In \cref{model} we formally present the model and definitions. In \cref{negative}, we present the analysis of the above-surveyed negative results (to Problem~1, to Problem~2, and to an additional more technical problem). In \cref{main}, we present \cref{intro-theorem}, which is our main result, and give a high-level overview of its proof. The proof itself is given in \cref{sec-r,sec-l,sec-b}, with some calculations relegated to the appendix.

\subsection{Further Related Work}
The literature on ``mechanism design from samples'' is preceded and inspired by the literature on prior-free and prior-independent mechanism design.
Early work in Algorithmic Mechanism Design has mainly focused on \emph{prior-free} mechanism design, aiming for either worst-case welfare approximation \cite[e.g.,][]{LehmannOS02} or worst-case revenue approximation with respect to some instance-specific benchmark (e.g., the best revenue when selling at least 2 items \citep{GoldbergHW01}).
For more results on prior-free mechanisms see, for example, Chapter 7 of \citet{HartlineBook}.

The standard economic model of revenue maximization assumes that the value of each player is drawn from a known prior distribution, and the seller aims to maximize her expected revenue for that prior \citep{Myerson81}. \citet{BK96} have presented a remarkable result, showing that a seller can attain at least his optimal revenue for buyers with values drawn i.i.d.\ from a regular distribution by using the VCG mechanism, as long as she can recruit one additional buyer (whose value is drawn independently from the same distribution).
This mechanism is \emph{prior-independent} in the sense that the mechanism does not depend on the priors, 
yet the approximation is obtained with respect to the optimal revenue for the specific distribution, even though this optimal revenue is unknown when the mechanism is run.
\citet{HartlineR09} have initiated the systematic study of such prior-independent mechanisms. For more results on prior-independent mechanisms see, for example, Chapter 5 of \citet{HartlineBook}.

A slightly less demanding model than the prior-independent model is the model in which the mechanism has access to samples from the unknown underlying distribution, with the benchmark still being the (unknown-to-the-mechanism) optimal revenue for that specific distribution. The current paper uses this model, and prior work in this model is surveyed in the introduction above.

A somewhat similar model where the auction is also chosen based on sampled data, which is studied in the learning literature, is an online-learning model where the mechanism designer can use information from prior auctions to on-line optimize the parameters of the next auction \citep{Cesa-BianchiGM15,WeedPR16}. The goal in this model is to optimize the overall performance.

The literature on ``mechanism design from samples'' restricts the dependence of the auction mechanism on the full details of the buyer's valuation distribution, by having it depend only on sample valuations drawn from this distribution. An alternative approach to restricting the dependence of the auction mechanism on the full details of the valuation distribution is by having it depend only on certain statistical measures of the valuation distribution, such as its mean, its variance, or its median \citep{AM13,ADMW13}.

\section{Preliminaries}\label{model}

\subsection{Model and Notation}

\paragraph{Distributions and Revenues}
We consider one seller and one buyer. The seller has one good for sale, which has no value for the seller if it is left unsold. The buyer has a private value (valuation, i.e., maximum willingness to pay) for the good, which is drawn from some
distribution $F$. For each real price $p$, the \emph{expected revenue} attained from posting price $p$ is thus simply $p\cdot\prob{v\sim F}{v\ge p}$. The highest possible expected revenue attainable from any price is denoted by $\OPT(F)$. The seller does not know the value of the buyer or the distribution $F$.

\paragraph{Empirical Revenue Maximization}
Given $n$ samples from $F$, the \emph{empirical distribution} over these $n$ samples is simply the uniform distribution over the samples, i.e., sample $i$ is drawn with probability~$\nicefrac{1}{n}$. The Empirical Revenue Maximization algorithm is given $n$ independent samples from $F$, computes the price $p$ that maximizes the expected revenue attained from the empirical distribution over the given $n$ samples, and posts this price. We denote by $\ERM(F,n)$ the expected revenue of the price computed by the ERM algorithm over a fresh draw from $F$ (independent from the $n$ samples used to compute the price posted by the ERM algorithm); note that this revenue is in expectation over both the $n$ samples and the fresh draw.

\paragraph{Quantile Space}
For our analysis, it would be convenient to reason about the possible values of the buyer using their \emph{quantiles}. The quantile of a value $v$ with respect to a distribution $F$ is $q(v)=q_F(v)=\prob{v\sim F}{v\ge p}\in[0,1]$.\footnote{Throughout this paper, we use many definitions that depend on the distribution $F$. To avoid clutter, we will omit the respective distribution $F$ from these notations when it is clear from context.} (So the revenue from posting a price $p$ is $p\cdot q(p)$.) Note that lower quantiles $q$ correspond to higher valuations $v$. We also define the inverse map, from quantiles back to values: for a quantile $q\in[0,1]$, the value corresponding to that quantile
with respect to a given atomless
distribution $F$
is denoted by $v(q)=v_F(q)$ (and is well defined since $F$ is atomless). We note that sampling a value $v\sim F$ is therefore equivalent to uniformly sampling a quantile $q\in[0,1]$ and then taking the value corresponding to that quantile $v=v(q)$.

\paragraph{Revenue Curve in Quantile Space}
The \emph{revenue curve (in quantile space)} that corresponds to an atomless value distribution $F$ is the mapping $r:[0,1]\rightarrow\RR$ from a quantile $q$ to the expected revenue $r(q)=v(q)\cdot q$ of posting the value $v(q)$ as the price. We note that the value function $v(\cdot)$ (and hence also the distribution $F$) can be recovered from the revenue curve via $v(q)=\frac{r(q)}{q}$, that is, $v(q)$ is precisely the slope of the line connecting the origin to the point $(q,r(q))$. We will at times write $\OPT(r)$ instead of $\OPT(F)$, write $v_r$ instead of $v_F$, etc. Note that $\OPT(r)=\max_{q\in[0,1]}r(q)$. An atomless distribution~$F$ is called \emph{regular} if its corresponding revenue curve (in quantile space) is concave.\footnote{While the more popular definition of regularity \citep{Myerson81} is phrased using \emph{virtual values}, these two standard definitions are known to be equivalent (in fact, the definition used here is more general as it also applies to nondifferentiable revenue curves). Indeed, it is well known that the derivative of the revenue curve at quantile $q$ is \citeauthor{Myerson81}'s virtual value at $v(q)$, and so the definition of regularity as the virtual-value function being increasing corresponds to the derivative of the revenue curve being decreasing, i.e., to the revenue curve being convex.}

\subsection{Additional Notation}

\begin{definition}[$e_2$]
Given a regular distribution $F$ with revenue curve $r$, we define $e_2:[0,1]^2\rightarrow\RR$, as follows.
\[e_2(q_1,q_2)=e_2^r(q_1,q_2)\eqdef\begin{cases}r\left(\arg\max_{q\in\{q_1,q_2\}}v(q)\right) & \max\bigl\{v(q_1),v(q_2)\bigr\}\geq 2\min\bigl\{v(q_1),v(q_2)\bigr\}, \\
r\left(\arg\min_{q\in\{q_1,q_2\}}v(q)\right) & \mbox{otherwise};\end{cases}\]
note that
\[\ERM(F,2)\eqdef \expect{(q_1,q_2)\sim U([0,1]^2)}{e_2(q_1,q_2)}.\]
\end{definition}
Namely, given a pair of quantiles $(q_1,q_2)$, if the value $\max\bigl\{v(q_1),v(q_2)\bigr\}$ is at least twice as large as the value $\min\bigl\{v(q_1),v(q_2)\bigr\}$ then $e_2(q_1,q_2)$ is the expected revenue of the price $\max\bigl\{v(q_1),v(q_2)\bigr\}$. Otherwise $e_2(q_1,q_2)$ is the expected revenue of the price $\min\bigl\{v(q_1),v(q_2)\bigr\}$.

During our analysis, it will be useful to work with revenue curves $r$ that are normalized so that $OPT(r)=1$. The following simple \lcnamecref{e2-homogeneous}, whose proof is given in \cref{homogeneity-proof} for completeness, justifies that this is without loss of generality.

\begin{lemma}\label[lemma]{e2-homogeneous}
For every regular distribution $F$ with revenue curve $r$, and for every $\alpha>0$, we have that\ \ \emph{(1)}~$v_{\alpha\cdot r}(q)=\alpha\cdot v_r(q)$ for every $q\in[0,1]$,\ \ \emph{(2)}~$e_2^{\alpha\cdot r}(q_1,q_2)=\alpha\cdot e_2^r(q_1,q_2)$ for every $q_1,q_2\in[0,1]$, and\ \ \emph{(3)}~$OPT(\alpha\cdot r)=\alpha\cdot OPT(r)$.
\end{lemma}

\subsection{A Single Sample}

In their paper, \citet{DRY10} show that a celebrated theorem by \citet{BK96} can be reinterpreted to imply that $\ERM$ guarantees one half of the optimal expected revenue for every regular distribution.

\begin{theorem}[\citealp{DRY10}]\label{one-half}
$\!\ERM(F,1)\!\ge\!\frac{1}{2}OPT(F)$ for every regular distribution~$F$. Furthermore, this is tight, i.e., the constant $\frac{1}{2}$ cannot be replaced with any larger constant in this statement.
\end{theorem}

\citet{DRY10} also give a direct simple proof for \cref{one-half} that does not use \citeauthor{BK96}'s result: recall that the quantile of a value drawn from $F$ is distributed uniformly in $[0,1]$. Therefore, the expected revenue by using a value drawn from $F$ as a price is precisely the integral of the revenue curve $r$, i.e., the area under the curve $r(\cdot)$. Since this curve is convex, the area under it is at least half of the height of the highest point on this curve. (This bound is tight: fixing any triangular revenue curve, there exists a sequence of regular distributions whose revenue curves uniformly converge to this triangular curve. Thus, revenue curves with areas under them that are arbitrarily close to one half of the height of their highest point can be constructed.) \cref{one-half} therefore follows since this height is precisely $OPT(r)=OPT(F)$.

As noted in the introduction, while the bound of one half from \cref{one-half} cannot be improved upon by any deterministic mechanism, \citet{FILS15} do manage to nonetheless construct, using one sample, a \emph{randomized} pricing mechanism that guarantees a strictly higher revenue of ${(0.5+5\times10^{-9})\cdot OPT(F)}>\frac{1}{2}\cdot OPT(F)$.

\section{Three Negative Results}\label{negative}

In this section we present the analysis of the three negative results surveyed in the introduction.

\begin{proposition}\label[proposition]{more-samples-bad}
There exists a regular distribution $F$ such that $\ERM(F,2)<\ERM(F,1)$.
\end{proposition}

\begin{proof}
The \emph{equal revenue distribution} is the distribution supported on $[1,\infty)$ with revenue $1$ for every posted price, that is, the distribution $G$ satisfying $\prob{v\sim G}{v\ge p}=\nicefrac{1}{p}$ for every $p\ge1$.

We take $F$ to be the equal revenue distribution, truncated at $v=10$ so that all of the mass of the equal revenue distribution $G$ at values $\ge10$ is uniformly respread in $F$ throughout a small interval $[10,10+\varepsilon]$. The corresponding revenue curve (in quantile space) $r(q)=r_F(q)$ climbs up almost linearly (with a very slight convex curvature, which tends to linear as $\varepsilon$ grows small) from $q=0$ (revenue $0$) until $q=0.1$ (revenue $1$), and continues at revenue $1$ from that point (i.e., for all quantiles $q>0.1$). For simplicity, we will perform our calculations by approximating the slightly curved convex climb of $r(q)$ in $[0,0.1]$ be a linear climb (conceptually corresponding to $\varepsilon$ tending to zero), that is:
\[r(q)=\begin{cases}10\cdot q & q \le 0.1, \\ 1 & \text{otherwise}.\end{cases}\]
It is easy (and standard) to see that this approximation will have a negligible effect on our calculations of $\ERM(F,2)$ and $\ERM(F,1)$, as its effect, for any quantile $q$, on either $r(q)$ or $v(q)$ is negligible. For this to indeed hold, in the definition of $e_2$ when defining the quantile chosen by the $\arg\min$ operator, we will henceforth break ties between $q_1,q_2$ that have $v(q_1)=v(q_2)$ (such distinct $q_1,q_2$ can only occur in the initial linear climb of the revenue curve) in favor of larger quantiles (that is, higher revenue), as in the slightly curved initial convex climb of the revenue curve that is approximated by this linear climb, the value of a larger quantile is slightly smaller that that of a smaller quantile.

We start by precisely calculating the expected revenue from posting a price computed by ERM given one sample:
\[\ERM(F,1)=\int_0^1r(q)\,dq=1-\frac{0.1}{2}=\nicefrac{19}{20}=0.95.\]
To calculate the expected revenue from posting a price computed by ERM given two samples, we note that the revenue will be nonoptimal (i.e., less than $1$) in precisely the following two cases:
\begin{itemize}
\item
Both samples are from a quantile $<0.1$. The expected revenue, conditioned on this case, is~$\expect{q_1,q_2\sim U([0,0.1])^2}{\max\{10\cdot q_1,10\cdot q_2\}}=\nicefrac{2}{3}$.
\item
One sample is from a quantile $q_1<0.1$ and the other is from a quantile $q_2>0.2$ In this case, since $v(q_1)\ge10=2\cdot5=2\cdot v(0.2)>2\cdot v(q_2)$, the price calculated by ERM given these two samples is $v(q_1)$, and so the expected revenue, conditioned on this case, is~$\expect{q_1\sim U([0,0.1])}{10\cdot q_1}=\nicefrac{1}{2}$.
\end{itemize}
Therefore, the expected revenue from posting a price computed by ERM given two samples is:
\begin{multline*}
\ERM(F,2)=\\
1-0.1^2\cdot\left(1-\expect{q_1,q_2\sim U([0,0.1])^2}{\max\{10q_1,10q_2\}}\right)-2\cdot0.1\cdot0.8\cdot\left(1-\expect{q_1\sim U([0,0.1])}{10q_1}\right)=\\
=1-0.1^2\cdot\nicefrac{1}{3}-2\cdot0.1\cdot0.8\cdot\nicefrac{1}{2}=\nicefrac{11}{12}=0.91\bar{6}.
\end{multline*}
And so, indeed, $\ERM(F,2)=\nicefrac{11}{12}<\nicefrac{19}{20}=\ERM(F,1)$, as required.
\end{proof}

\begin{proposition}\label[proposition]{switch}
There exist regular distributions $F$ and $G$ such that $\ERM(F,1)>\ERM(G,1)$ and $\ERM(F,2)<\ERM(G,2)$.
\end{proposition}

\begin{proof}
An intuitive strategy for proving \cref{switch} is as follows: recall from the proof of \citet{DRY10} of \cref{one-half} that for distributions with ``triangular'' revenue curves, the expected revenue from posting a price computed by ERM given one sample, is precisely one half of optimal, since the area beneath the revenue curve is precisely one half of the maximum value of this curve. Nonetheless, when posting a price computed by ERM given two samples, it is not hard to observe that different triangular revenue curves result in different expected revenues. So, we will take $F$ and $G$ to be two distributions corresponding to triangular revenue curves with $OPT(F)=OPT(G)$, such that $\ERM(F,2)<\ERM(G,2)$. Since both revenue curves are triangular, we have that $\ERM(F,1)=\ERM(G,1)$. By slightly perturbing the revenue curve of $F$ in a way that increases the area under this curve (causing $\ERM(F,1)$ to increase) while only slightly changing $\ERM(F,2)$, we obtain that for this perturbed $F$ it still holds that $\ERM(F,2)<\ERM(G,2)$, but at the same time it also holds that $\ERM(F,1)>\ERM(G,1)$, as required. We omit the full details as \cref{switch} also follows from \cref{increase-area}, whose more subtle proof we give below.
\end{proof}

Recall from the proof of \citet{DRY10} for the single-sample case, that for any distribution $F$, the expected revenue $\ERM(F,1)$ is the integral of the revenue curve $r_F(\cdot)$ with respect to the uniform measure over quantiles. In the case of two samples, $\ERM(F,2)$ can still be expressed as an appropriate integral of the revenue curse $r_F(\cdot)$, with two main caveats: first, since two samples are involved, this integral is no longer with respect to the uniform measure on quantiles, and second, since the probability of using the price that corresponds to a certain quantile, i.e., the probability that ERM given two samples chooses this price, in fact depends in quite a delicate manner on the distribution $F$ through the value function $v_F(\cdot)$, the measure with respect to which this integral is defined is itself intricately dependent on the distribution $F$. To drive this point home, we present the following surprising observation:

\begin{proposition}\label[proposition]{increase-area}
There exist two regular distributions $F$ and $G$ with $r_F(q)>r_G(q)$ for all\footnote{Equivalently, $F$ first-order stochastically dominates $G$.} $q\in(0,1)$ (so $\int r_F(q)>\int r_G(q)$ with respect to any measure), such that $\ERM(F,2)<\ERM(G,2)$.
\end{proposition}

\begin{proof}[Proof of \cref{increase-area}]
We will choose distributions $F$ and $G$ for which $OPT(F)=OPT(G)=1$.
We take $G$ to be the distribution function corresponding to the ``triangular'' revenue curve $r_G(q)=q$
and take $F$ to be the distribution function corresponding to the following ``quadrilateral'' revenue curve, which is obtained from $r_G$ by adding a slight ``bump'' (while maintaining convexity) at $q=0.1$ (see \cref{fig-quadrilateral}):
\begin{figure}[t]
\centering
\begin{tikzpicture}[font=\allfigfont]
\begin{axis}[ymin=0,enlargelimits=false,axis x line*=bottom,axis y line*=left,xtick={0,0.1,0.571429,1}, xticklabels={$0$,$0.1$,$t\approx0.571$,$1$},ytick={0,.22,1},yticklabels={$0$,$0.22$,$1$},clip=false,scale=\allfigscale]
\addplot[domain=0.0001:1, white] {x};
\draw[green, ultra thick] (axis cs: 0,0) -- (axis cs: 1,1);
\draw[blue, ultra thick] (axis cs: 0,0) -- (axis cs: 0.1,0.22) -- (axis cs: 1,1);
\draw[dashed] (axis cs: 1,1) -- (axis cs: 0,1);
\draw[dashed] (axis cs: 1,1) -- (axis cs: 1,0);
\draw[dashed] (axis cs: 0.1,0.22) -- (axis cs: 0,0.22);
\coordinate (q) at (axis cs:0.1,0.22);
\path let \p1 = (q) in coordinate (qq) at (\x1, 0);
\draw[dashed] (q) -- (qq);
\draw[dashed] (q) -- (axis cs: 0,0);
\draw[red,thick] ($(q)!.25!(qq) + (axis cs:-0.025,0.01)$) -- ++(axis cs:0.05,0) ++(axis cs:0,-.02) -- ++(axis cs:-0.05,0);
\draw[red,thick] ($(q)!.75!(qq) + (axis cs:-0.025,0.01)$) -- ++(axis cs:0.05,0) ++(axis cs:0,-.02) -- ++(axis cs:-0.05,0);
\draw[dashed] (axis cs: 0.571429,0.628571) -- (axis cs: 0,0);
\draw[dashed] (axis cs: 0.571429,0.628571) -- (axis cs: 0.571429,0);
\end{axis}
\end{tikzpicture}
\caption{The revenue curves $r_F$ (blue) and $r_G$ (green) corresponding to the distributions $F$ and $G$ from \cref{increase-area}. The quantile $t$ is the one that satisfies $2\cdot v_F(t)=v_F(0.22)$; that is, $t$ is the $q$-coordinate of the intersection point of $r_F$ and the ray from $(0,0)$ with slope exactly half that of the ray from $(0,0)$ to $(0.1,0.22)$.}\label{fig-quadrilateral}
\end{figure}
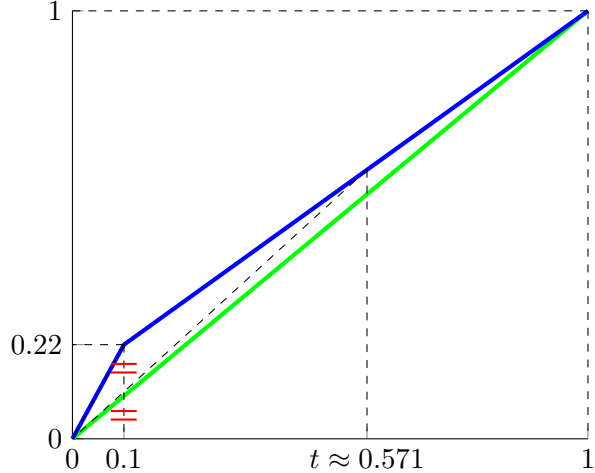%
\[
r_F(q)=\begin{cases} 2.2\cdot q & q\le0.1 \\ 1-0.78\cdot\frac{1-q}{0.9} & \text{otherwise}.\end{cases}
\]
By construction $\ERM(F,1)=\int r_F(q)\,dq>\int r_G(q)\,dq = \ERM(G,1)$.
It remains to show that $\ERM(F,2)<\ERM(G,2)$.\footnote{As in our proof of \cref{more-samples-bad}, in order to avoid atoms in $F$ and in $G$, we uniformly spread the mass of the atom at the highest value~$v$ in the support of each distribution throughout the interval $[v,v+\varepsilon]$, corresponds to replacing the first linear climb in each of their revenue curves with a slightly curved convex climb.}

We start with $G$ (which corresponds to the uniform distribution over $[1,1+\varepsilon]$ for negligible epsilon). By a simple calculation we have for one sample that $\ERM(G,1)=\nicefrac{1}{2}$ and for two samples that $\ERM(G,2)=\expect{q_1,q_2\in U([0,1])}{\max\bigl\{r(q_1),r(q_2)\bigr\}}=\nicefrac{2}{3}$. We now move on to $F$; the intuition behind our construction is that by adding the ``bump'' at $0.1$, while we do increase revenue by slightly raising $\max\bigl\{r(q_1),r(q_2)\bigr\}$ above $\nicefrac{2}{3}$, we in fact decrease the revenue by ``confusing'' the ERM algorithm and causing it to choose the higher price of the two samples, which corresponds to the lower revenue $\min\bigl\{r(q_1),r(q_2)\bigr\}$, in some cases, such as whenever the quantile of one of the samples is below $0.1$ and the quantile of the other is above $t\approx0.571$ (since in this case $v\bigl(\min\{q_1,q_2\}\bigr)>2\cdot v\bigl(\max\{q_1,q_2\}\bigr)$). As it turns out, the latter effect dominates the former one, causing an overall decrease in expected revenue compared to that of $F$. In \cref{s:calc-increase-area}, we calculate and show that indeed $\ERM(F,2)<0.651<\ERM(G,2)$, thereby completing the proof.
\end{proof}

\section{Main Result}\label{main}

In this section, we phrase and prove our main result.

\begin{theorem}\label{more-than-half}
For every regular distribution $F$, we have that $\ERM(F,2)>0.509\cdot\OPT(F)$.
\end{theorem}

\cref{more-than-half} is the first part of \cref{intro-theorem} from the introduction. Combining \cref{more-than-half} with \cref{one-half}, we obtain our main monotonicity result (the second part of \cref{intro-theorem} from the introduction).

\begin{corollary}\label[corollary]{monotonicity}
$\inf_{\text{regular $F$}}\frac{\ERM(F,2)}{OPT(F)}>\inf_{\text{regular $F$}}\frac{\ERM(F,1)}{OPT(F)}$.
\end{corollary}

As noted in the introduction, proving \cref{more-than-half} turns out to be considerably more technically challenging than may have been expected. One hint as to why was already given in \cref{increase-area}. To prove \cref{more-than-half}, we lower-bound the integral that defines $\ERM(F,2)$ by estimating it over three domains: first, we estimate this integral conditioned upon the two samples corresponding to prices lower than the ideal posted price; second, we estimate this integral condition upon the two samples corresponding to prices higher than the ideal posted price; and finally, we estimate this integral conditioned upon the two samples falling on opposite sides of the ideal posted price. Estimating each of these integrals is quite involved. For the first two domains, we manage to show that ERM guarantees quantifiably more than one half of the optimal revenue, while unfortunately for the third domain, we do not manage to show that ERM guarantees even half of the optimal revenue, forcing us to estimate the integral for the first two domains tightly enough to enable us to argue that the losses in this third domain could be charged to the gains in the first two domains. To balance the charging argument, we must utilize quite a few observations in each domain, and furthermore estimate the integral in the first and last domains functions of the quantile of the ideal price. Putting all of these together, we manage to show an overall guarantee of at least $0.509$ of $OPT(F)$ --- quite close to one half, but nonetheless strictly bounded away from one half (and still greater by orders of magnitude than the guarantee of $(0.5+5\times10^{-9})\cdot OPT(F)$ that \citet{FILS15} show for their randomized mechanism for one sample).

One possible way to approach \cref{more-than-half} (and more generally, \cref{all-n}) could have been to try and identify the worst-case distributions for $2$ samples (and more generally, for $n$ samples), i.e., distributions $F$ for which $\frac{\ERM(F,2)}{\OPT(F)}$ (and more generally, $\frac{\ERM(F,n)}{\OPT(F)}$) is smallest, and then to calculate this fraction for such distributions $F$. Indeed, recall that this is how \citet{DRY10} have proved \cref{one-half}: they have identified the distributions with triangular revenue curves as the worst-case distributions for one sample, and showed that for these distributions this fraction equals one half. Unfortunately, following this path, even for two samples, turns out quite elusive. In this vain, hoping that some distribution with a triangular revenue curve continues to be a worst-case distribution, Huang, Mansour, and Roughgarden (personal communication, 2015) have identified the single distribution with triangular revenue curve for which this fraction is smallest among all distributions with triangular revenue curves, and have calculated this fraction for this distribution (incidentally, it turns out to be considerably higher than our lower bound from \cref{more-than-half}). Unfortunately, in light of our proof of \cref{increase-area}, it seems that one can show that for some distribution with a quadrilateral revenue curve (created, similarly to the proof of \cref{increase-area}, by adding a small ``bump'' to the ``left edge'' of the triangular revenue curve identified by Huang, Mansour, and Roughgarden), this fraction turns out to be smaller than for the distribution identified by Huang, Mansour, and Roughgarden, and hence we conclude that no distribution with a triangular revenue curve is a worst-case distribution for ERM for two samples. As we do not know how to identify the worst-case distributions for ERM, even for two samples, our analysis must bound the fraction $\frac{\ERM(F,2)}{\OPT(F)}$ for arbitrary regular distributions.

In the remainder of this section, we survey the high-level ideas behind the proof of \cref{more-than-half}, whose full proof is given in \cref{sec-l,sec-r,sec-b}, with some calculations relegated to \cref{s:calc}. Formally, we partition the set of pairs of quantiles as follows:

\begin{definition}[$L$; $R$; $B$]
For a regular distribution $F$ with revenue curve $r$, we define\footnote{In case of multiple revenue-maximizing quantiles, we may pick $\qo$ arbitrarily among them.} $\qo\eqdef\arg\max_{q\in[0,1]}r(q)$ and partition $[0,1]^2$ into three sets:
\begin{enumerate}
\item
$R \eqdef\bigl\{(q_1,q_2)\in[0,1]^2~\big|~\min\{q_1,q_2\}\ge \qo\bigr\}$,
\item
$L\eqdef\bigl\{(q_1,q_2)\in[0,1]^2~\big|~\max\{q_1,q_2\}< \qo\bigr\}$, and
\item
$B \eqdef [0,1]^2\setminus(L\cup R)$.
\end{enumerate}
\end{definition}

To prove \cref{more-than-half}, we will lower-bound the expected revenue of a price chosen by the ERM algorithm given two samples, conditioned upon these two samples belonging to each of the sets $R$, $L$, and $B$.

\subsection{Both Samples Below the Ideal Price (\texorpdfstring{\hspace{-.2ex}$\mathbfit{R}$}{R})}
In \cref{sec-r}, we lower-bound the expected revenue, conditioned upon both samples being lower than the ideal price $\qo$:

\begin{lemma}\label[lemma]{l:right}
For every regular $F$,
\[
\expect{(q_1,q_2)\sim U(R)}{e_2(q_1,q_2)} \ge
\begin{cases}
\!\left(\dfrac{1}{3}-\dfrac{1}{4(1-\qo)}-\dfrac{1}{2(1-\qo)^2}-\dfrac{\log\qo}{2(1-\qo)^3}\right)\!\cdot\OPT(F)
& \qo\!\geq\!\nicefrac{2}{3},\\
\!\left(\dfrac{2}{3}-\dfrac{1}{(1-\qo)^3}\left(\dfrac{2}{9}-\left(\dfrac{\qo}{4}\right)^2+\dfrac{1}{3}\left(\dfrac{\qo}{2}\right)^3+\dfrac{\log\nicefrac{2}{3}}{2}\right)\!\right)\!\cdot\OPT(F)
\!\!& \qo\!<\!\nicefrac{2}{3}.
\end{cases}
\]
\end{lemma}

As noted above, it would not have sufficed to simply bound the expectation on the left-hand side of the above inequality away from one half of $\OPT$, as we will have to charge the losses below one-half-of-$\OPT$ of our lower bound for the case where the quantiles are in $B$ (given in \cref{l:both} below), which depend on $\qo$, to the gains above one-half-of-$\OPT$ from \cref{l:right} (and from \cref{l:left} below).
To prove \cref{l:right}, we first lower-bound, for each possible quantile $q$, the expected revenue of the price computed by ERM conditioned upon $\min\{q_1,q_2\}=q$. Recall that since both quantiles are in $R$, we have that $q$ is the quantile of the ``better'' sample, with higher expected revenue. Denoting by $t(q)$ the threshold value of $\max\{q_1,q_2\}$ (see \cref{fig-r-t})
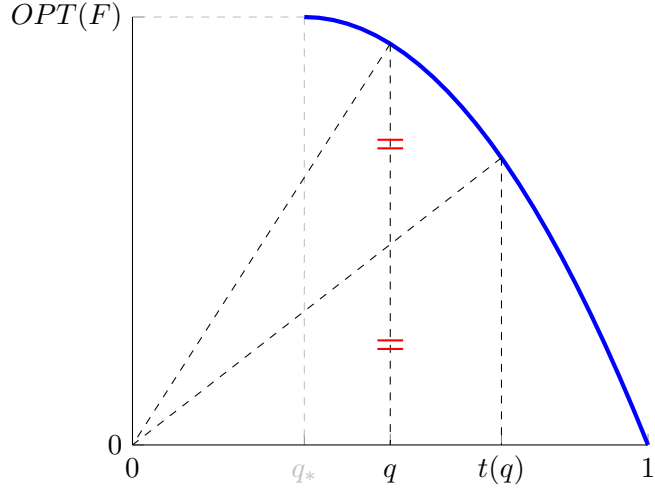
\begin{figure}[t]
\centering
\begin{tikzpicture}[font=\allfigfont]
\begin{axis}[xmin=0,enlargelimits=false,axis x line*=bottom,axis y line*=left,xtick={0,0.3333,0.5,.7157,1}, xticklabels={$0\vphantom{(}$,$\textcolor{lightgray}{\qo\vphantom{(}}$,$q\vphantom{(}$,$t(q)$,$1\vphantom{(}$}, ytick={0,1}, yticklabels={$0$, $OPT(F)$}, clip=false,scale=\allfigscale]
\draw[dashed,color=lightgray] (axis cs: 1/3,1) -- (axis cs: 1/3,0);
\draw[dashed,color=lightgray] (axis cs: 1/3,1) -- (axis cs: 0,1);
\addplot[domain=0.3333:1, blue, ultra thick] {(4/3*(1-x) - (1-x)^2)*9/4};
\coordinate (q) at (axis cs:0.5,.9375);
\path let \p1 = (q) in coordinate (qq) at (\x1, 0);
\draw[dashed] (q) -- (qq);
\draw[dashed] (q) -- (axis cs: 0,0);
\draw[red,thick] ($(q)!.25!(qq) + (axis cs:-0.025,0.01)$) -- ++(axis cs:0.05,0) ++(axis cs:0,-.02) -- ++(axis cs:-0.05,0);
\draw[red,thick] ($(q)!.75!(qq) + (axis cs:-0.025,0.01)$) -- ++(axis cs:0.05,0) ++(axis cs:0,-.02) -- ++(axis cs:-0.05,0);
\coordinate (t) at (axis cs:{(0.5625+sqrt(0.5625^2+27/4))/4.5}, {0.9375*(0.5625+sqrt(0.5625^2+27/4))/4.5});
\draw[dashed] (t) -- (axis cs:0,0);
\path let \p1 = (t) in coordinate (tt) at (\x1, 0);
\draw[dashed] (t) -- (tt);
\end{axis}
\end{tikzpicture}
\caption{Pictorial definition of $t(q)$
from the proof of \cref{l:right}. When $r$ is continuous with $r(1)=0$, then $t(q)$ is the $q$-coordinate of the intersection point of $r$ and the ray from $(0,0)$ with slope exactly half that of the ray from $(0,0)$ to $\bigl(q,r(q)\bigr)$.}\label{fig-r-t}
\end{figure}%
so that the ERM algorithm chooses the price that corresponds to quantile $q$ (the price that attains better expected revenue among the two samples), it turns out that we are in a win-win situation: if $t(q)$ is ``close'' to $q$, then conditioned upon $\min\{q_1,q_2\}=q$, there is a high probability that the higher quantile lies above $t(q)$, causing the ``better'' sample (with quantile $q$) to be chosen; conversely, the farther $t(q)$ is from $q$, the larger $r\bigl(t(q)\bigr)$ is, i.e., the revenue curve $r$ decreases quite moderately between $\qo$ and $t(q)$, and so even when the ``worse'' sample is chosen, the revenue is still reasonably high. The full proof of \cref{l:right} is given in \cref{sec-r}, with some calculations relegated to \cref{s:calc-right}.

\subsection{Both Samples Above the Ideal Price (\texorpdfstring{\hspace{-.2ex}$\mathbfit{L}$}{L})}

In \cref{sec-l}, we lower-bound the expected revenue, conditioned upon both samples being higher than the ideal price $\qo$:

\begin{lemma}\label[lemma]{l:left}
For every regular $F$,
\[\expect{(q_1,q_2)\sim U(L)}{e_2(q_1,q_2)} \ge 0.528\cdot\OPT(F).\]
\end{lemma}

To survey the proof of \cref{l:left}, we define:
\begin{itemize}
\item $E_2^L(r)=\expect{(q_1,q_2)\sim U(L)}{e_2(q_1,q_2)}$ --- the expected revenue of ERM,
\item $M_2^L(r)=\expect{(q_1,q_2)\sim U(L)}{\max\bigl\{r(q_1),r(q_2)\bigr\}}$ --- the expected revenue had we always picked the ``better'' sample (the sample with higher expected revenue).
\end{itemize}
The proof of \cref{l:left} is based on a coupling of $E_2^L$ and $M_2^L$, giving that
\begin{itemize}
\item
$E_2^L(r)\geq \frac{3}{4}\cdot M_2^L(r)$, with equality if and only if $r$ is constant\footnote{That is, $\forall q:~ r(q)= \OPT(r)$.}.
\end{itemize}
The proof of this inequality is similar to the proof of \cref{l:right}, and is achieved by lower-bounding, for each possible quantile $q$, the expected revenue of the price computed by ERM conditioned upon $\max\{q_1,q_2\}=q$. Recall that since both quantiles are in $L$, we have that $q$ is the quantile of the ``better'' sample, with higher expected revenue. Denoting by $t(q)$ the threshold value of $\min\{q_1,q_2\}$ (see \cref{fig-l-t})
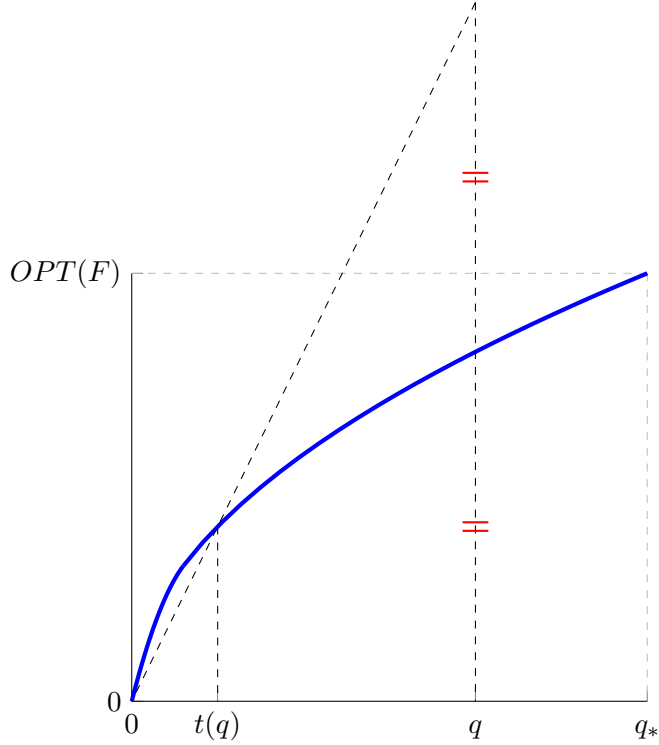
\begin{figure}[t]
\centering
\begin{tikzpicture}[font=\allfigfont]
\begin{axis}[enlargelimits=false,axis x line*=bottom,axis y line*=left,xtick={0,0.1666,.6666,1}, xticklabels={$0\vphantom{(}$,$t(q)$,$q\vphantom{(}$,$\qo\vphantom{(}$}, ytick={0,1}, yticklabels={$0$, $OPT(F)$}, clip=false,scale=\allfigscale]
\draw[dashed,color=lightgray] (axis cs: 1,1) -- (axis cs: 1,0);
\draw[dashed,color=lightgray] (axis cs: 1,1) -- (axis cs: 0,1);
\addplot[domain=0:0.1, blue, ultra thick] {-15.8114*x^2+4.74342*x};
\addplot[domain=0.1:1, blue, ultra thick] {sqrt(x)};
\coordinate (twiceq) at (axis cs: 2/3,{2*sqrt(2/3)});
\path let \p1 = (twiceq) in coordinate (qq) at (\x1, 0);
\draw[dashed] (twiceq) -- (qq);
\draw[dashed] (twiceq) -- (axis cs: 0,0);
\draw[red,thick] ($(twiceq)!.25!(qq) + (axis cs:-0.025,0.01)$) -- ++(axis cs:0.05,0) ++(axis cs:0,-.02) -- ++(axis cs:-0.05,0);
\draw[red,thick] ($(twiceq)!.75!(qq) + (axis cs:-0.025,0.01)$) -- ++(axis cs:0.05,0) ++(axis cs:0,-.02) -- ++(axis cs:-0.05,0);
\coordinate (t) at (axis cs:1/6, {sqrt(1/6)});
\path let \p1 = (t) in coordinate (tt) at (\x1, 0);
\draw[dashed] (t) -- (tt);
\end{axis}
\end{tikzpicture}
\caption{Pictorial definition of $t$
from the proof of \cref{l:left}. $t(q)$ is the $q$-coordinate of the rightmost intersection point of $r$ and the ray from $(0,0)$ with slope exactly twice that of the ray from $(0,0)$ to $\bigl(q,r(q)\bigr)$.}\label{fig-l-t}
\end{figure}%
so that the ERM algorithm chooses the price that corresponds to quantile $q$ (the price that attains better expected revenue among the two samples), we show that we are once again in a win-win situation: if $t(q)$ is ``close'' to $q$, then conditioned upon $\max\{q_1,q_2\}=q$, there is a high probability that the lower quantile lies below $t(q)$, causing the ``better'' sample (with quantile~$q$) to be chosen; conversely, the farther $t(q)$ is from $q$, the larger $r\bigl(t(q)\bigr)$ is, i.e., the revenue curve $r$ decreases quite moderately between $\qo$ and $t(q)$, and so even when the ``worse'' sample is chosen, the revenue is still reasonably close to that of the ``better'' sample.

Unlike in the proof of \cref{l:right}, the proof of \cref{l:left} requires some additional case-analysis beyond this point. A simple calculation shows that
\begin{itemize}
\item
$M_2^L(r)\geq \frac{2}{3}\cdot \OPT(r)$, with equality if and only if $r$ is linear\footnote{That is, $\exists m > 0$, so that  $\forall q:~ r(q)=m\cdot q$.}.
\end{itemize}
The observations in the two ``bullets'' above lead to the following win-win situation:
\begin{itemize}
\item Either $r$ is ``far from linear'', and then $M_2^L(r) >\bigl(\frac{2}{3}+\varepsilon\bigr)\cdot OPT(r)$, giving
\[E_2^L(r)\geq \tfrac{3}{4}\cdot M_2^L(r) > \tfrac{3}{4}\cdot\bigl(\tfrac{2}{3}+\varepsilon\bigr)\cdot OPT(r)> \tfrac{1}{2}\cdot OPT(r),\]
\item  or $r$ is ``far from constant'', and then $E_2^L(r)\geq\bigl(\frac{3}{4}+\varepsilon\bigr)\cdot M_2^L(r)$, giving
\[E_2^L(r)\geq\bigl(\tfrac{3}{4}+\varepsilon\bigr)\cdot M_2^L(r) > \bigl(\tfrac{3}{4}+\varepsilon\bigr)\cdot\tfrac{2}{3}\cdot OPT(r)>\tfrac{1}{2}\cdot OPT(r).\]
\end{itemize}

The full proof of \cref{l:left} is given in \cref{sec-l}, with some calculations relegated to \cref{s:calc-left}.

\subsection{One Sample on Each Side of the Ideal Price (\texorpdfstring{\hspace{-.2ex}$\mathbfit{B}$}{B})}

In \cref{sec-b}, we lower-bound the expected revenue, conditioned upon one of the two samples being lower than the ideal price $\qo$ and the other being higher than the ideal price $\qo$:

\begin{lemma}\label[lemma]{l:both}
For every regular $F$,
\[\expect{(q_1,q_2)\sim U(B)}{e_2(q_1,q_2)} \ge \frac{1}{2}\cdot\left(1 - \left(\frac{\qo}{2\cdot(1+\qo)}\right)^2\right)\cdot\OPT(F).\]
\end{lemma}

As already mentioned above, whenever $\qo>0$, the right-hand side of the above inequality (which deteriorates as $\qo$ grows) is strictly less than the one-half-of-$\OPT(F)$ guarantee of ERM given one sample, so when proving \cref{more-than-half} we will have to charge this loss to the gains over one-half-of-$\OPT(F)$ that we obtained in the lower bounds of \cref{l:right,l:left}. The correctness of \cref{l:both} follows from the following \lcnamecref{trr}.

\begin{sublemma}\label[sublemma]{trr}
Let $m\in[0,1]$, let $T:[0,1]\rightarrow[m,1]$  be a monotone nonincreasing function, and for every $i\in\{1,2\}$ let $r_i:[0,1]\rightarrow[0,1]$ be a monotone nondecreasing and concave function s.t.\ $r_i(x)\ge x$ for every $x\in[0,1]$. For every $(x,y)\in [0,1]^2$, define $G(x,y)\in[0,1]$ as follows:
\[
G(x,y)\eqdef\begin{cases}r_2(y) & r_2(y)\ge T(x), \\ r_1(x) & \mbox{otherwise.}\end{cases}\]
Then $\expect{(x,y)\sim U([0,1]^2)}{G(x,y)}\ge \frac{1+m}{2} - \nicefrac{1}{2}\cdot(\frac{1+m}{2})^2$.
\end{sublemma}

The proof of \cref{trr} is given in \cref{sec-b}. We will now show how \cref{l:both} indeed follows from this \lcnamecref{trr}.

\begin{proof}[Proof of \cref{l:both}]
By \cref{e2-homogeneous}, we may assume without loss of generality that $\OPT(r)=1$.
By symmetry of $e_2$, it is enough to prove the claim w.r.t.\ $(q_1,q_2)\sim B_1$, where $B_1 \eqdef \{(q_1,q_2)\in B \mid q_1<q_2\}$.
We define:
\begin{itemize}
\item
$r_1:[0,1]\rightarrow[0,1]$ by $x\mapsto r(x\cdot\qo)$,
\item
$r_2:[0,1]\rightarrow[0,1]$ by $y\mapsto r\bigl(1-y\cdot(1-\qo)\bigr)$,
\item
$T:[0,1]\rightarrow[0,1]$ by $x\mapsto r\left(\inf\bigl\{q\ge\qo ~\middle|~ v(x\cdot\qo) \ge 2v(q)\bigr\}\right)$. (See \cref{fig-b-t}.)
\begin{figure}[t]
\centering
\subfigure[Pictorial definition of $T$ from the proof of \cref{l:both}. When $r$ is continuous with $r(1)=0$, then $T(x)$ is the $y$-coordinate of the intersection point of $r$ and the ray from~$(0,0)$ with slope exactly half that of the ray from $(0,0)$ to $\bigl(x\cdot\qo,r(x\cdot\qo)\bigr)$.]{\label{fig-b-t}
\begin{tikzpicture}[font=\doublefigfont]
\begin{axis}[xmin=0,enlargelimits=false,axis x line*=bottom,axis y line*=left,xtick={0,.5,.6666,1}, xticklabels={$0$,$x\cdot\qo\vphantom{0}$,$\textcolor{lightgray}{\qo\vphantom{0}}$,$1$}, ytick={0,1}, yticklabels={$0$, $OPT(F)$},clip=false,xscale=\doublefigxscale,yscale=\allfigscale]
\draw[dashed,color=lightgray] (axis cs: 2/3,1) -- (axis cs: 2/3,0);
\draw[dashed,color=lightgray] (axis cs: 2/3,1) -- (axis cs: 0,1);
\addplot[domain=0:0.6666, blue, ultra thick] {(4/3*x - x^2)*9/4};
\addplot[domain=0.6666:1, blue, ultra thick] {(4/3*(2*x-2/3) - (2*x-2/3)^2)*9/4};
\coordinate (q) at (axis cs:0.5,.9375);
\path let \p1 = (q) in coordinate (qq) at (\x1, 0);
\draw[dashed] (q) -- (qq);
\draw[dashed] (q) -- (axis cs: 0,0);
\draw[red,thick] ($(q)!.25!(qq) + (axis cs:-0.025,0.01)$) -- ++(axis cs:0.05,0) ++(axis cs:0,-.02) -- ++(axis cs:-0.05,0);
\draw[red,thick] ($(q)!.75!(qq) + (axis cs:-0.025,0.01)$) -- ++(axis cs:0.05,0) ++(axis cs:0,-.02) -- ++(axis cs:-0.05,0);
\coordinate (t) at (axis cs:{(11.0625+sqrt(14.37890625))/18}, {0.9375*(11.0625+sqrt(14.37890625))/18});
\draw[dashed] (axis cs:0,0) -- (t);
\path let \p1 = (t) in coordinate (tt) at (\x1, 0);
\draw[dashed] (tt) -- (t);
\draw[decorate,decoration={brace,amplitude=10pt},yshift=0pt] (tt) -- (t) node [midway,anchor=east,inner sep=0,outer sep=0,xshift=-11pt] {$T(x)$};
\end{axis}
\end{tikzpicture}
}\quad
\subfigure[Pictorial definition of $m$ from the proof of \cref{l:both}. Concavity of~$r$ implies that the graph of the decreasing part of~$r$ must be completely contained in the shaded region. $m$ is the smallest possible value of $T(1)$ (which is attained when~$r$ is the smallest possible within the shaded region).]{\label{fig-b-m}
\begin{tikzpicture}[font=\doublefigfont]
\begin{axis}[xmin=0,enlargelimits=false,axis x line*=bottom,axis y line*=left,xtick={0,.6666,1}, xticklabels={$0$,$\qo\vphantom{0}$,$1$}, ytick={0,1}, yticklabels={$0$, $OPT(F)$},clip=false,xscale=\doublefigxscale,yscale=\allfigscale]
\addplot[domain=0:0.6666, blue, ultra thick] {(4/3*x - x^2)*9/4};
\addplot[domain=0.6666:1, blue, ultra thick] {(4/3*(2*x-2/3) - (2*x-2/3)^2)*9/4};
\fill[blue,opacity=0.25] (axis cs: 2/3,1) -- (axis cs: 1,1) -- (axis cs: 1,0) -- cycle;
\coordinate (q) at (axis cs:2/3,1);
\path let \p1 = (q) in coordinate (qq) at (\x1, 0);
\draw[dashed] (q) -- (qq);
\draw[dashed] (q) -- (axis cs: 0,0);
\draw[dashed,color=lightgray] (q) -- (axis cs: 0,1);
\draw[red,thick] ($(q)!.25!(qq) + (axis cs:-0.025,0.01)$) -- ++(axis cs:0.05,0) ++(axis cs:0,-.02) -- ++(axis cs:-0.05,0);
\draw[red,thick] ($(q)!.75!(qq) + (axis cs:-0.025,0.01)$) -- ++(axis cs:0.05,0) ++(axis cs:0,-.02) -- ++(axis cs:-0.05,0);
\coordinate (t) at (axis cs:{(11.25+sqrt(18.5625))/18}, {0.75*(11.25+sqrt(18.5625))/18});
\draw[dashed] (axis cs:0,0) -- (t);
\path let \p1 = (t) in coordinate (tt) at (\x1, 0);
\draw[dashed] (tt) -- (t);
\draw[decorate,decoration={brace,amplitude=17pt,mirror},yshift=0pt] (tt) -- (t) node [midway,anchor=west,inner sep=0,outer sep=0,xshift=19pt,overlay] {$T(1)$};
\coordinate (m) at (axis cs:0.8, 0.75*0.8);
\path let \p1 = (m) in coordinate (mm) at (\x1, 0);
\draw[dashed] (mm) -- (m);
\draw[decorate,decoration={brace,amplitude=10pt},yshift=0pt] (mm) -- (m) node [midway,anchor=east,inner sep=0,outer sep=0,xshift=-11pt] {$m$};
\end{axis}
\end{tikzpicture}
}
\caption{Reduction from \cref{l:both} to \cref{trr}.}\label{fig-b}
\end{figure}%
\end{itemize}
We notice that under these definitions, we have (where $G(x,y)$ is defined as in \cref{trr}) that
\[\expect{(q_1,q_2)\sim U(B_1)}{e_2(x,y)} =
\expect{(x,y)\sim U([0,1]^2)}{G(x,y)}.\]

We note that the lower-bound that we obtain by applying \cref{trr}, as is, to the above definitions of $r_1,r_2,T$, is $\nicefrac{3}{8}$, which for $\qo<1$ is worse than the guarantee that we are attempting to prove in \cref{l:both}. The main idea that drives our improved revenue guarantee in \cref{l:both} is to bound $m$, the minimum value attained by $T$, away from $0$.
Since the intersection of $y=\frac{x}{2\qo}$ (the line defined by the points $(0,0)$ and $(\qo,\nicefrac{1}{2})$) and
$y=1-\frac{x-\qo}{1-\qo}$ (the line defined by the points $(\qo,1)$ and $(1,0)$) is $(\frac{2\qo}{1+\qo},\frac{1}{1+\qo})$, then defining $m\eqdef\frac{1}{1+\qo}$,
we note that by monotonicity and convexity of $r|_{[\qo,1]}$ and since $r(\qo)=1$, the range of $T$ is in fact $[m,1]$. (See \cref{fig-b-m}.)

Applying \cref{trr} with this value of $m$, we obtain that
\begin{multline*}
\expect{(q_1,q_2)\sim U(B_1)}{e_2(x,y)} =
\expect{(x,y)\sim U([0,1]^2)}{G(x,y)} \ge\\
\frac{1+m}{2} - \frac{1}{2}\cdot\left(\frac{1+m}{2}\right)^2 =
\frac{1}{2}\left(1-\left(\frac{\qo}{2(1+\qo)}\right)^2\right),
\end{multline*}
as required. The calculation justifying the last equality is detailed in \cref{s:calc-both}.
\end{proof}

\subsection{Completing the Proof of Theorem~\ref{more-than-half}}

\cref{more-than-half} follows from combining \cref{l:right,l:left,l:both}.

\begin{proof}[Proof of \cref{more-than-half}]
By definition,
\begin{multline*}
\ERM(F,2)=(1-\qo)^2\cdot\expect{(q_1,q_2)\sim U(R)}{e_2(q_1,q_2)}+\qo^2\cdot\expect{(q_1,q_2)\sim U(L)}{e_2(q_1,q_2)}+\\*
2\qo(1-\qo)\cdot\expect{(q_1,q_2)\sim U(B)}{e_2(q_1,q_2)}.
\end{multline*}
In \cref{s:calc-combine}, we substitute the respective lower bounds from \cref{l:right,l:left,l:both} for each of the above summands, and calculate and show that indeed $\ERM(F,2)>0.509\cdot\OPT(F)$, thereby completing the proof of \cref{more-than-half}.
\end{proof}

\section{Both Samples Below the Ideal Price (\texorpdfstring{\hspace{-.2ex}$\mathbfit{R}$}{R}): Proof of Lemma~\ref{l:right}}\label{sec-r}

\begin{proof}[Proof of \cref{l:right}]

By \cref{e2-homogeneous}, we may assume without loss of generality that $\OPT(r)=1$. We define:
\[E_2^R(r)=\expect{(q_1,q_2)\sim U(R)}{e_2(q_1,q_2)}.\]

Let $(q_1,q_2)\sim U\bigl([\qo,1]^2\bigr)$.
The random variable $\min\{q_1,q_2\}$ has density function $\mu(q)=2\frac{1-q}{(1-\qo)^2}$ (see
\cref{l:min}
in \cref{s:minmax}).
For $q\in [\qo,1]$, define
\begin{align*}
E_2^R(r \vert q) &\eqdef \expect{}{e_2(q_1,q_2) ~\middle\vert~ \min\{q_1,q_2\}\!=\!q }.
\end{align*}
Note that
\begin{align*}
E_2^R(r) &= \expect{q\sim\mu}{E_2^R(r \vert q)}.
\end{align*}

Let $q\in[0,1]$. Conditioned\footnote{This event has $0$ probability. However, since the pair $\bigl(\min\{q_1,q_2\},\max\{q_1,q_2\}\bigr)$ has an appropriate joint density function, this conditioning is meaningful.}  on $\min\{q_1,q_2\}=q$, we have that $e_2(q_1,q_2)$ is equal to $r(q)=
\max\bigl\{r(q_1),r(q_2)\bigr\}$ exactly when
$2v\bigl(\max\{q_1,q_2\}\bigr) \leq v(q)$. Thus, we define $t(q)$ as the threshold value that determines when $e_2(q_1,q_2)=r(q)=
\max\bigl\{r(q_1),r(q_2)\bigr\}$.
\begin{definition}[$t(x)$; see \cref{fig-r-t} on \cpageref{fig-r-t}]\label[definition]{r-t}
For $q\in [\qo,1]$, we define
\[t(q) \eqdef \sup\left\{x \geq q ~\middle|~ 2v(x) > v(q)\right\}.\]
For $q=1$, define $t(1)\eqdef1$, and when $\qo=0$, define $t(0)\eqdef0$.
\end{definition}

It is not hard to verify the following properties of the function $t(q)$.
\begin{sublemma}[Properties of $t$]\label[sublemma]{l:r-propt}\leavevmode
\begin{enumerate}
\item $e_2(q_1,q_2)=
\max\bigl\{r(q_1),r(q_2)\bigr\}$
if and only if $\max\{q_1,q_2\}\geq t\bigl(\min\{q_1,q_2\}\bigr)$.
\item For all $q$: $t(q)\in [q,2q]$ (due to monotonicity of $r$).
\item $t$ is monotone nondecreasing (due to monotonicity of $v$).
\item $t$ is continuous (due to continuity of $r$).
\item\label[property]{l:r-propt-v} $v(t(q)) = \frac{v(q)}{2}$, whenever $t(q)\ne0$.
\end{enumerate}
\end{sublemma}

The following lemma relates $E_2^R(r\vert q)$ with $r(q)$,
which is the expected revenue obtained by always choosing the ``better'' sample when $\min\{q_1,q_2\}=q$.
\begin{sublemma}\label[sublemma]{l:coupInc}
For all $q$:
\[E_2^R(r\vert q) \geq \frac{r(q)}{1-q}\left(\frac{t(q)^2}{4q} - \frac{3}{4}t(q) + 1-\frac{q}{2}\right).\]
\end{sublemma}
\begin{proof}
Conditioned
on $\min\{q_1,q_2\}=q$, the random variable $x\eqdef\max\{q_1,q_2\}$ is distributed uniformly in $[q,1]$.
Thus,
\begin{align*}
E_2^R(r\vert q) &= \prob{}{x \geq t(q)}\cdot r(q) + \prob{}{x < t(q)}\cdot\expect{x}{r(x) ~\Big\vert~ q\leq x < t(q)}\\
                     &\geq \prob{}{x \geq t(q)}\cdot r(q) + \prob{}{x < t(q)}\cdot\left(\frac{r\bigl(t(q)\bigr)+r(q)}{2}\right) \tag{by concavity of $r$}\\
                     &=\prob{}{x \geq t(q)}\cdot r(q) + \prob{}{x < t(q)}\cdot\left(\frac{\frac{t(q)}{2q}r(q)+r(q)}{2}\right) \tag{by \crefpart{l:r-propt}{v}: $\frac{r\bigl(t(q)\bigr)}{t(q)}=v\bigl(t(q)\bigr)=\frac{v(q)}{2}=\frac{r(q)}{2q}$}\\
                     &= \left(\frac{1-t(q)}{1-q}\right)\cdot r(q) + \left(\frac{t(q)-q}{1-q}\right)\cdot\left(\frac{\frac{t(q)}{2q}r(q)+r(q)}{2}\right)\\
                     &=\frac{r(q)}{1-q}\left(\frac{t(q)^2}{4q} - \frac{3}{4}t(q) + 1-\frac{q}{2}\right).\tag*{\qedhere}
\end{align*}
\end{proof}

\begin{sublemma}\label[sublemma]{c:dec}
For every $q\in[0,1]$:
\[
E_2^R(r\vert q)\geq
\begin{cases}
r(q)\left(1-\frac{1}{16}\cdot\frac{q}{1-q}\right)  &q\leq\nicefrac{2}{3},\\
r(q)\left(\frac{1}{2}+\frac{1}{4q}\right) &q>\nicefrac{2}{3}.
\end{cases}
\]
\end{sublemma}
\begin{proof}
By \cref{l:coupInc},
\[E_2^R(r\vert q) \geq \frac{r(q)}{1-q}P_q\bigl(t(q)\bigr).\]
Where $P_q(x)\eqdef\frac{x^2}{4q} - \frac{3}{4}x + 1-\frac{q}{2}$ is a quadratic polynomial
with minimum at $x_{\min}\eqdef\frac{3}{2}q$.
Since $t(q)\leq 1$, it follows that
\[
P_q\bigl(t(q)\bigr)\geq
\begin{cases}
P_q(\frac{3}{2}q)  &\frac{3}{2}q\leq 1,\\
P_q(1) &\frac{3}{2}q>1.
\end{cases}
\]
Now, $\frac{3}{2}q \leq 1$ if and only if $q\leq\frac{2}{3}$.
The conclusion follows since
\begin{align*}
\frac{r(q)}{1-q}P_q\left(\tfrac{3}{2}q\right) &= r(q)\left(1-\frac{1}{16}\cdot\frac{q}{1-q}\right), \mbox{\quad and}\\
\frac{r(q)}{1-q}P_q\left(1\right) &=r(q)\left(\frac{1}{2}+\frac{1}{4q}\right).\qedhere
\end{align*}
\end{proof}

With \cref{c:dec} in hand, we are ready to finish the proof of \cref{l:right}.
We will use the following simple fact that follows from the concavity of $r$:
\begin{equation}\label{r-concav}
\mbox{$\forall q\in[\qo,1]$:}\quad r(q)\geq\frac{1-q}{1-\qo};
\end{equation}
indeed, $\ell(q)=\frac{1-q}{1-\qo}$
is an affine function satisfying $\ell(\qo)=1=r(\qo)$, and $\ell(1)=0\leq r(1)$.

In \cref{s:calc-right}, we use \cref{c:dec,r-concav} to calculate and show that indeed $E_2^R(r)$ is lower-bounded as in the statement of \cref{l:right}, thereby completing the proof of this \lcnamecref{l:right}.
\end{proof}

\section{Both Samples Above the Ideal Price (\texorpdfstring{\hspace{-.2ex}$\mathbfit{L}$}{L}): Proof of Lemma~\ref{l:left}}\label{sec-l}

\paragraph{Proof of \cref{l:left}.}
When proving \cref{l:right}, we had no choice but to carry on the dependence on $\qo$ throughout the proof. Luckily, in the case where both samples are above the ideal price, we may normalize $\qo$ to equal $1$ without loss of generality. Indeed, the resulting distribution has all values multiplied by $\qo$ (similarly to the proof of \cref{e2-homogeneous}), and so the choice of ERM between the values corresponding to any two quantiles is unchanged. Furthermore, as before, by \cref{e2-homogeneous} we may assume without loss of generality that $\OPT(r)=1$.

We begin by coupling $M_2^L$ and $E_2^L$.
Let $(q_1,q_2)\sim U\bigl([0,1]^2\bigr)$.
The random variable $\max\{q_1,q_2\}$ has density function $\mu(q)=2q$ (see \crefpart{l:max}{density} in \cref{s:minmax}).
For $q\in [0,1]$, define
\begin{align*}
E_2^L(r \vert q) &\eqdef \expect{}{e_2(q_1,q_2) ~\vert~ \max\{q_1,q_2\}\!=\!q },\\
M_2^L(r \vert q) &\eqdef \expect{}{
\max\bigl\{r(q_1),r(q_2)\bigr\}~\vert~ \max\{q_1,q_2\}\!=\!q } = r(q).
\end{align*}
Note that
\begin{align*}
E_2^L(r) &= \expect{q\sim\mu}{E_2^L(r \vert q)},\\
M_2^L(r) &= \expect{q\sim\mu}{M_2^L(r \vert q)}.
\end{align*}

Next, we relate $E_2^L(r \vert q)$ with $M_2^L(r \vert q)$.
Let $q\in[0,1]$. Conditioned on $\max\{q_1,q_2\}=q$,\footnote{Similarly to \cref{sec-r}, this event has $0$ probability. However, since the pair $\bigl(\min\{q_1,q_2\},\max\{q_1,q_2\}\bigr)$ has an appropriate joint density function, this conditioning is meaningful.} we have that $e_2(q_1,q_2)$ is equal to $r(q)=
\max\bigl\{r(q_1),r(q_2)\bigr\}$ exactly when
$v\bigl(\min\{q_1,q_2\}\bigr) < 2v(q)$. Thus, we define $t(q)$ as the threshold value that determines when $e_2(q_1,q_2)=r(q)=
\max\bigl\{r(q_1),r(q_2)\bigr\}$.

\begin{definition}[$t(x)$; see \cref{fig-l-t} on \cpageref{fig-l-t}]\label[definition]{l-t}
For $q\in [0,1]$, we define
\[t(q) \eqdef \inf\bigl\{x\le q ~\big|~ v(x) \geq 2v(q)\bigr\}.\]
If $q=0$ or $\bigl\{x ~\big|~ v(x) \geq 2v(q)\bigr\}=\emptyset$, then we define $t(q)\eqdef 0$.
\end{definition}

It is not hard to verify the following properties of the function $t(q)$.
\begin{sublemma}[Properties of $t$]\label[sublemma]{l:l-propt}\leavevmode
\begin{enumerate}
\item $e_2(q_1,q_2)=
\max\bigl\{r(q_1),r(q_2)\bigr\}$ if and only if $\min\{q_1,q_2\} \geq t\bigl(\max\{q_1,q_2\}\bigr)$.
\item\label[property]{l:l-propt-half} For all $q$: $t(q)\leq \frac{q}{2}$ (due to monotonicity of $r$).
\item\label[property]{l:l-propt-inc} $t$ is monotone nondecreasing (due to monotonicity of $v$).
\item $t$ is continuous (due to continuity of $r$).
\item\label[property]{l:l-propt-v} $v(t(q)) = 2v(q)$, whenever $t(q)\neq 0$.
\end{enumerate}
\end{sublemma}

The following \lcnamecref{e2-parab} relates $M_2^L(r\vert q)$ and $E_2^L(r\vert q)$.
\begin{sublemma}\label[sublemma]{e2-parab}
For all $q$:
\[E_2^L(r\vert q) \geq r(q)\cdot\left(\left(\frac{t}{q}\right)^2 - \frac{t}{q} + 1\right).\]
\end{sublemma}
\begin{proof}
Conditioned on $\max\{q_1,q_2\}=q$,
the random variable $x\eqdef\min\{q_1,q_2\}$ is distributed uniformly in $[0,q]$.
Thus,
\begin{align*}
E_2^L(r\vert q) &= \prob{}{x \ge t}\cdot r(q) + \prob{}{x < t}\cdot\expect{x\sim U([0,1])}{r(x) ~\middle\vert~ x < t}\\
                     &\geq \prob{}{x \ge t}\cdot r(q) + \prob{}{x<t}\cdot\expect{x\sim U([0,1])}{\frac{r(t)}{t}x ~\middle\vert~ x < t}\tag{by concavity, $\forall x\leq t:~\frac{r(x)}{x}\ge\frac{r(t)}{t}$}\\
                     &=\prob{}{x \ge t}\cdot r(q) + \prob{}{x<t}\cdot\frac{r(t)}{2}\\
                     &=\prob{}{x \ge t}\cdot r(q) + \prob{}{x<t}\cdot\frac{t}{q}r(q)\tag{by \crefpart{l:l-propt}{v}, $\forall q:~\frac{r(t)}{t}=v(t)=2v(q)=2\frac{r(q)}{q}$}\\
                     &=\frac{q-t}{q}\cdot r(q) + \frac{t}{q}\cdot\frac{t}{q}r(q)\\
                     &=r(q)\cdot\left(\left(\frac{t}{q}\right)^2 - \frac{t}{q} + 1\right).\qedhere
\end{align*}
\end{proof}

Since for every $y$ the value of the polynomial $y^2-y+1$ is at least $\frac{3}{4}$, we obtain the following \lcnamecref{c:coup}:
\begin{sublemma}\label[sublemma]{c:coup}
For every $q\in[0,1]$: $E_2^L(r\vert q)\geq \frac{3}{4} r(q)=\frac{3}{4}M_2^L(r\vert q)$.
\end{sublemma}

We next quantify the distinction between revenue curves~$r$ that are ``close'' to constant,
and revenue curves~$r$ that are ``close'' to linear.
To this end, fix some $\delta\in[0,1]$ that we will later pick in a way that maximizes the revenue. We distinguish between two types of revenue curves~$r$:\footnote{It is possible also to optimize over ``$\nicefrac{1}{2}$'' in this partition of revenue curves, that is, for some carefully chosen $q_m,\delta\in[0,1]$, to distinguish between revenue curves $r$ satisfying $r(q_m)\geq (1+\delta)\cdot q_m$ and revenue curves $r$ satisfying $r(q_m)<(1+\delta)\cdot q_m$. While this does give a slightly improved lower bound, we avoid doing as it would add even more details to the analysis.}
\begin{enumerate}
\item $r$ satisfies $r(\nicefrac{1}{2})\geq (1+\delta)\cdot\nicefrac{1}{2}$ ($r$ is ``far'' from linear).
\item $r$ satisfies $r(\nicefrac{1}{2})< (1+\delta)\cdot\nicefrac{1}{2}$ ($r$ is ``far'' from constant).
\end{enumerate}

\begin{sublemma}[The Case in which  $r$ is ``Far from Linear'']\label[sublemma]{l:inc_case1-m}
If $r$ satisfies $r(\nicefrac{1}{2})\geq (1+\delta)\cdot\nicefrac{1}{2}$, then
\[M_2^L(r)\geq \frac{2}{3} + \frac{\delta}{4}.\]
\end{sublemma}
\begin{proof}
Let $\ell_1(q)$ be the affine function that satisfies $\ell_1(0)=0$ and $\ell_1(\nicefrac{1}{2})=(1+\delta)\cdot\nicefrac{1}{2}$,
and let $\ell_2(q)$ be the affine function that satisfies $\ell_2(\nicefrac{1}{2})=(1+\delta)\cdot\nicefrac{1}{2}$ and $\ell_2(1)=1$.
By concavity of $r$ and since $r(1)=1$,
\[
r(q) \ge \begin{cases} \ell_1(q) &q\in[0,\frac{1}{2}],\\ \ell_2(q) &q\in[\frac{1}{2},1]. \end{cases}
\]
Therefore,
\begin{align*}
M_2^L(r) &= \expect{q\sim\mu}{M_2^L(r \vert q)}\\
            &= \expect{q\sim\mu}{r(q)}\\
            &= \prob{}{q\leq\nicefrac{1}{2}}\cdot\expect{q}{r(q) ~\middle|~ q\leq\nicefrac{1}{2}}+
            \prob{}{q>\nicefrac{1}{2}}\cdot\expect{q}{r(q) ~\middle|~ q>\nicefrac{1}{2}}\\
            &\ge \prob{}{q\leq\nicefrac{1}{2}}\cdot\expect{q}{\ell_1(q) ~\middle|~ q\leq\nicefrac{1}{2}}+
            \prob{}{q>\nicefrac{1}{2}}\cdot\expect{q}{\ell_2(q) ~\middle|~ q>\nicefrac{1}{2}}\\
            &=\prob{}{q\leq\nicefrac{1}{2}}\cdot\ell_1\bigl(\expect{q}{q ~\middle|~ q\leq\nicefrac{1}{2}}\bigr)+
            \prob{}{q>\nicefrac{1}{2}}\cdot\ell_2\bigl(\expect{q}{q ~\middle|~ q>\nicefrac{1}{2}}\bigr) \tag{by linearity of expectation}\\
            &= \nicefrac{1}{4}\cdot\ell_1(\nicefrac{1}{3}) + \nicefrac{3}{4}\cdot\ell_2(\nicefrac{7}{9})\tag{by \cref{l:max}(\labelcref{l:max-le},\labelcref{l:max-ge}) in \cref{s:minmax}}\\
            &= \frac{1+\delta}{12} + \frac{21+6\delta}{36}\\
            &= \frac{2}{3} + \frac{\delta}{4}.\qedhere
\end{align*}
\end{proof}

\noindent Combining \cref{l:inc_case1-m,c:coup} yields:
\begin{sublemma}\label[sublemma]{l:inc_case1}
If $r$ satisfies $r(\nicefrac{1}{2})\geq (1+\delta)\cdot\nicefrac{1}{2}$, then
\[E_2^L(r)\geq\frac{1}{2} +\frac{3\delta}{16}.\]
\end{sublemma}

\begin{sublemma}[The Case in which $r$ is ``Far from Constant'']\label[sublemma]{l:inc_case2}
If $r$ satisfies $r(\nicefrac{1}{2})< (1+\delta)\cdot\nicefrac{1}{2}$, then
\begin{align*}
E_2^L(r) &\geq \frac{64}{3}\gamma^5-8\gamma^4+\frac{8}{3}\gamma^3-\frac{2}{3}\gamma^2-\gamma+\frac{2}{3},
\end{align*}
where $\gamma\eqdef\frac{\delta}{1+\delta}$. (See \cref{fig-l-g}.)
\begin{figure}[t]
\centering
\begin{tikzpicture}[font=\allfigfont]
\pgfmathsetmacro{\dd}{0.3};
\pgfmathsetmacro{\gg}{\dd/(1+\dd)};
\begin{axis}[ymin=0,enlargelimits=false,axis x line*=bottom,axis y line*=left,xtick={0,0.0992,\gg,0.5,1}, xticklabels={$0\vphantom{(}$,$t(1)$,$\gamma\vphantom{(}$,$\nicefrac{1}{2}\vphantom{(}$,$\qo=1\vphantom{(}$},ytick={0,.5,0.65,1},yticklabels={$0$,$\nicefrac{1}{2}$,$\nicefrac{1}{2}\cdot(1+\delta)$,$OPT(F)=1$},clip=false,scale=\allfigscale]
\addplot[domain=0.0001:1, blue, ultra thick] {x^.7};
\draw[dashed] (axis cs: 1,1) -- (axis cs: 0,\dd);
\fill[blue,opacity=0.25] (axis cs: 0,0) -- (axis cs: 0,\dd) -- (axis cs: .5,{.5*(1+\dd)}) -- (axis cs: .5,.5) -- cycle;
\draw[dashed] (axis cs: 1,2) -- (axis cs: 1,0);
\draw[red,thick] (axis cs: 1-0.025,1.5) -- ++(axis cs:0.05,0) ++(axis cs:0,-.02) -- ++(axis cs:-0.05,0);
\draw[red,thick] (axis cs: 1-0.025,.5) -- ++(axis cs:0.05,0) ++(axis cs:0,-.02) -- ++(axis cs:-0.05,0);
\draw[dashed] (axis cs: 1,2) -- (axis cs: 0,0);
\draw[dashed] (axis cs: 1,1) -- (axis cs: 0,1);
\draw[dashed] (axis cs: .5,.5) -- (axis cs: 0,.5);
\draw[dashed] (axis cs: .5,{.5*(1+\dd)}) -- (axis cs: 0,{.5*(1+\dd)});
\draw[dashed] (axis cs: .5,0) -- (axis cs: .5,{.5*(1+\dd)});
\draw[dashed] (axis cs: \gg,2*\gg) -- (axis cs: \gg,0);
\coordinate (t) at (axis cs:{2^(-1/.3)},{2*2^(-1/.3)});
\path let \p1 = (t) in coordinate (tt) at (\x1, 0);
\draw[dashed] (t) -- (tt);
\end{axis}
\end{tikzpicture}
\caption{Pictorial definition of $\gamma$ from \cref{l:inc_case2}. The assumption that $r(\nicefrac{1}{2})<(1+\delta)\cdot\nicefrac{1}{2}$ implies that the graph of $r$ restricted to $[0,\nicefrac{1}{2}]$ must be completely contained in the shaded region. $\gamma$ is the largest possible value of $t(1)$ (which is attained when $r$ is the largest possible within the shaded region).}\label{fig-l-g}
\end{figure}
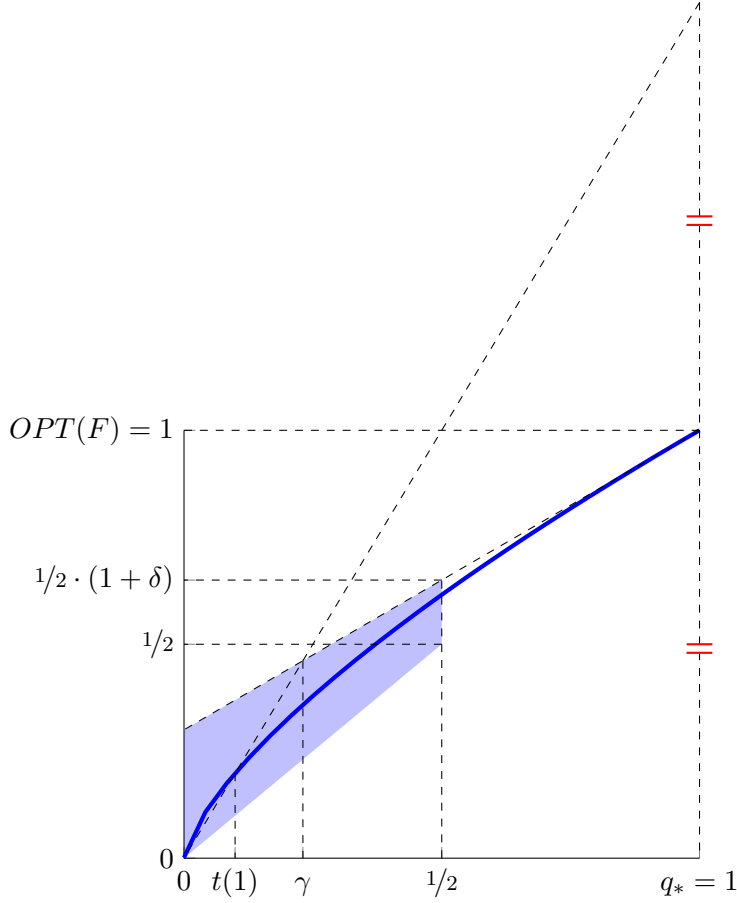%
\end{sublemma}
\begin{proof}
We first show that
\[\forall q:~t(q)\leq\min\{\nicefrac{q}{2},\gamma\}.\]
That $t(q)\leq\frac{q}{2}$ follows from \crefpart{l:l-propt}{half}.
To see that $t(q)\leq\gamma$ for all $q$, it suffices to show that $t(1)\leq\gamma$
(since $t(q)$ is monotone nondecreasing by \crefpart{l:l-propt}{inc}).
Since $r(1)=1$, we have $v(1)=1$, and so by \crefpart{l:l-propt}{v}, $t(1)$ is the $q$-coordinate of the intersection point of the line $y=2q$ and $r(q)$.
Let $\ell(q)$ be the affine function that satisfies $\ell(\nicefrac{1}{2})=(1+\delta)\cdot\nicefrac{1}{2}$ and $\ell(1)=1$.
By concavity of $r$, for all $q\in[0,\nicefrac{1}{2}]$, $r(q)\leq\ell(q)$. Therefore, since by \crefpart{l:l-propt}{half}, $t(1)\leq\nicefrac{1}{2}$,
it follows that $t(1)$ is at most the $q$-coordinate of the intersection point of the lines $y=2q$,
and $\ell(q)$. A simple calculation shows that the latter $q$-coordinate is $\frac{\delta}{1+\delta}=\gamma$,
and so $t(1)\leq\gamma$. (See \cref{fig-l-g}.) Therefore, we have shown that
\[\forall q:~t(q)\leq\min\{\nicefrac{q}{2},\gamma\},\]
as required.

Therefore, by \cref{e2-parab},
\begin{align}\label{l-case2-ege}
E_2^L(r)
	&\geq \expect{q}{\left(\left(\frac{t}{q}\right)^2 - \frac{t}{q} + 1\right)\cdot r(q)}\\
	&\geq \expect{q}{\left(\left(\frac{t}{q}\right)^2 - \frac{t}{q} + 1\right)\cdot q}\tag{by concavity and since $r(1)\!=\!1$, $\forall q:~r(q)\geq q$}\\
	&\geq\prob{}{q\leq 2\gamma}\cdot\expect{q}{\tfrac{3}{4}q ~\middle|~ q\leq2\gamma}+
           \prob{}{q > 2\gamma}\cdot\expect{q}{\left(\left(\frac{\gamma}{q}\right)^2-\frac{\gamma}{q}+1\right)\cdot q ~\middle|~  q > 2\gamma}.\tag{since when $q>2\gamma$, we have $\frac{t}{q}\le\frac{\gamma}{q}<\frac{1}{2}$, and since $y^2-y+1$ is decreasing in $[0,\frac{1}{2}]$}
\end{align}
In \cref{s:calc-left}, we calculate and show that the left-hand side of \cref{l-case2-ege} is lower-bounded by the right-hand side of the inequality in the statement of \cref{l:inc_case2}, thereby completing the proof of this \lcnamecref{l:inc_case2}.
\end{proof}

With \cref{l:inc_case1,l:inc_case2} in hand, the proof of \cref{l:left} follows;
indeed, taking $\delta\eqdef0.15117$ yields $\min\left\{\frac{1}{2} +\frac{3\delta}{16},\frac{64}{3}\gamma^5-8\gamma^4+\frac{8}{3}\gamma^3-\frac{2}{3}\gamma^2-\gamma+\frac{2}{3}\right\}\approx0.528344$ and therefore
\begin{center}
For every nondecreasing regular $r$, ~$E_2^L(r)\geq 0.528\cdot\OPT(r)$.
\end{center}
which concludes the proof of \cref{l:left}.
\qed

\section{One Sample on Each Side of the Ideal Price (\texorpdfstring{\hspace{-.2ex}$\mathbfit{B}$}{B}): Proof of Lemma~\ref{l:both}}\label{sec-b}

We complete the proof of \cref{l:both} by proving \cref{trr}, from which we have shown above that \cref{l:both} follows.

\begin{proof}[Proof of \cref{trr}]
Let $\hat{M}\eqdef\expect{y\sim U([0,1])}{r_2(y) \mid r_2(y)\!\ge\!m}$ and $p_2 \eqdef \prob{y\sim U([0,1])}{r_2(y)\!<\!m}$. We first claim that for every $x\in[0,1]$,
\begin{equation}\label{gy-ge}
\expect{y\sim U([0,1])}{G(x,y)}\ge \begin{cases} r_1(x) & r_1(x) \le \hat{M}, \\ p_2\cdot r_1(x) + (1-p_2)\cdot \hat{M} & \mbox{otherwise.}\end{cases}
\end{equation}

To prove \cref{gy-ge}, for every $x\in[0,1]$ we first put $M(x)\eqdef\expect{y\sim U([0,1])}{r_2(y) \mid r_2(y)\!\ge\!T(x)}$.
We note that since $M(x)$ is increasing in $T(x)$, we have that $M(x)\ge \hat{M}$. We now fix $x\in[0,1]$. By definition and since $M(x)\ge\hat{M}$, we have that
\begin{align*}
\expect{y\sim U([0,1])}{G(x,y)}
&=
\prob{}{r_2(y)\!<\!T(x)}\cdot r_1(x) + \prob{}{r_2(y)\!\ge\!T(x)}\cdot M(x) \ge \\
&\ge \prob{}{r_2(y)\!<\!T(x)}\cdot r_1(x) + \prob{}{r_2(y)\!\ge\!T(x)}\cdot \hat{M}.
\end{align*}
For the first part of \cref{gy-ge}, note that if $r_1(x)\le\hat{M}$, then
\begin{align*}
\expect{y\sim U([0,1])}{G(x,y)}
&\ge
\prob{}{r_2(y)\!<\!T(x)}\cdot r_1(x) + \prob{}{r_2(y)\!\ge\!T(x)}\cdot \hat{M} \ge \\
&\ge\prob{}{r_2(y)\!<\!T(x)}\cdot r_1(x) + \prob{}{r_2(y)\!\ge\!T(x)}\cdot r_1(x) = r_1(x).
\end{align*}
For the second part of \cref{gy-ge}, note that if $r_1(x)>\hat{M}$, then for any convex combination of $r_1(x)$ and $\hat{M}$, increasing the relative weight of $\hat{M}$ can only reduce the total value of the convex combination. Therefore,
\begin{align*}
\expect{y\sim U([0,1])}{G(x,y)}
&\ge
\prob{}{r_2(y)\!<\!T(x)}\cdot r_1(x) + \prob{}{r_2(y)\!\ge\!T(x)}\cdot \hat{M} \ge \\
&\ge\prob{}{r_2(y)\!<\!m}\cdot r_1(x) + \prob{}{r_2(y)\!\ge\!m}\cdot \hat{M} = \\
&\ge p_2\cdot r_1(x) + (1-p_2)\cdot \hat{M},
\end{align*}
concluding the proof of \cref{gy-ge}.

Let $p_1 \eqdef \prob{x\sim U([0,1])}{r_1(x)\!\le\!\hat{M}}$. By \cref{gy-ge}, we have that
\begin{multline}\label{gxy-ge}
\expect{(x,y)\sim U([0,1]^2)}{G(x,y)} = \expect{x\sim U([0,1])}{\expect{y\sim U([0,1])}{G(x,y)}} = \\
p_1\cdot\expect{x}{\expect{y}{G(x,y)} ~\big|~ r_1(x) \le \hat{M}} +
(1-p_1)\cdot\expect{x}{\expect{y}{G(x,y)} ~\big|~ r_1(x) > \hat{M}} \ge \\
p_1 \cdot \expect{x}{r_1(x) \mid r_1(x) \le \hat{M}} +
(1-p_1)\cdot \left(p_2\cdot \expect{x}{r_1(x) \mid r_1(x) > \hat{M}}+ (1-p_2)\cdot \hat{M}\right) \ge \\
p_1 \cdot \nicefrac{\hat{M}}{2} + (1-p_1)\cdot \left(p_2\cdot\frac{1+\hat{M}}{2}+ (1-p_2)\cdot \hat{M}\right),
\end{multline}
where the last inequality is by monotonicity and concavity of $r_1$ and since $r_1(1)=1$.

Since $\hat{M}\le 1$, we have that $\hat{M} \le \frac{1+\hat{M}}{2}$.
By concavity of $r_1$ and since $r_1(1)=1$, we note that $p_1 \le \hat{M}$. By both of these and by \cref{gxy-ge}, we have that
\begin{multline*}
\expect{(x,y)\sim U([0,1]^2)}{G(x,y)} ~~\ge~~ p_1 \cdot \nicefrac{\hat{M}}{2} + (1-p_1)\cdot \left(p_2\cdot\frac{1+\hat{M}}{2}+ (1-p_2)\cdot \hat{M}\right) ~~\ge \\
\ge~~ p_1 \cdot \nicefrac{\hat{M}}{2} + (1-p_1)\cdot \hat{M} ~~\ge~~ \hat{M}\cdot\nicefrac{\hat{M}}{2} + (1-\hat{M})\cdot\hat{M} ~~=~~ \hat{M} - \nicefrac{\hat{M}^2}{2}.
\end{multline*}

Finally, by monotonicity and concavity of $r_2$ and since since $r_2(1)=1$, we have that $\hat{M}\ge\frac{1+m}{2}$. Noting that the function $z-\nicefrac{z^2}{2}$ is increasing in $[0,1]$, we therefore have that
\[\expect{(x,y)\sim U([0,1]^2)}{G(x,y)}  \ge \hat{M} - \frac{\hat{M}^2}{2} \ge \frac{1+m}{2}-\frac{1}{2}\cdot\left(\frac{1+m}{2}\right)^2,\]
as required.
\end{proof}

\section*{Acknowledgments}

Yannai Gonczarowski is supported by the Adams
Fellowship Program of the Israel Academy of Sciences and Humanities;
his work is supported by ISF grant 1435/14 administered by the
Israeli Academy of Sciences, by Israel-USA Bi-national Science
Foundation (BSF) grant number 2014389, and by
the European Research Council (ERC) under the European
Union's Horizon 2020 research and innovation programme (grant
agreement No 740282). 
The research of Yishay Mansour was supported
in part by a grant from the Israel Science
Foundation, a grant from United States-Israel Binational
Science Foundation (BSF), and the Israeli Centers of Research Excellence (I-CORE)
program (Center  No.\ 4/11).
The research of Shay Moran is supported by the National Science Foundation under agreement No.\ CCF-1412958 and by the Simons Foundations.
We thank Zhiyi Huang and Tim Roughgarden for stimulating conversations.

\bibliographystyle{plainnat}
\bibliography{two-samples}

\appendix

\section{Omitted Calculations}\label{s:calc}

\subsection{Calculation Omitted from the Proof of Proposition~\ref{increase-area}}\label{s:calc-increase-area}

\begin{align*}
\ERM(F,2)&<
0.1^2\cdot\expect{q_1.q_2\sim U([0,0.1])^2}{\max\{r(q_1),r(q_2)\}}+\\
&\qquad\qquad+0.9^2\cdot\expect{q_1,q_2\sim U([0.1,1])^2}{\max\{r(q_1),r(q_2)\}}+\\*
&\qquad\qquad+2\cdot(t-0.1)\cdot0.1\cdot\expect{q\sim U([0.1,t])}{r(q)}+\\*
&\qquad\qquad+2\cdot0.1\cdot(1-t)\cdot\expect{q\sim U([0,0.1])}{r(q)}=\\
&=0.1^2\cdot\frac{2}{3}\cdot0.22+\\*
&\qquad\qquad+0.9^2\cdot\bigl(0.22+\frac{2}{3}\cdot(1-0.22)\bigr)+\\*
&\qquad\qquad+2\cdot0.1\cdot(t-0.1)\cdot\bigl(0.22+\frac{1}{2}\cdot(r(t)-0.22)\bigr)+\\*
&\qquad\qquad+2\cdot0.1\cdot(1-t)\cdot\frac{1}{2}\cdot0.22<\\
&\qquad\qquad<0.651<\ERM(G,2).\tag*{\qed}
\end{align*}

\subsection{Calculation Omitted from the Proof of Lemma~\ref{l:right}}\label{s:calc-right}

Consider first the case in which $\qo\geq\nicefrac{2}{3}$.
\begin{align*}
E_2^R(r) &= \expect{q}{E_2^R(r\vert q)}\\
       &\geq \expect{q}{r(q)\left(\frac{1}{2}+\frac{1}{4q}\right)}\tag{by \cref{c:dec}}\\
       &\geq \expect{q}{\left(\frac{1-q}{1-\qo}\right)\left(\frac{1}{2}+\frac{1}{4q}\right)}\tag{by \cref{r-concav}}\\
       &=\int_{\qo}^{1}\left(\frac{1-q}{1-\qo}\right)\left(\frac{1}{2}+\frac{1}{4q}\right)\left(2\frac{1-q}{(1-\qo)^2}\right)dq\tag{by
\cref{l:min}
in \cref{s:minmax}}\\
       &=\int_{\qo}^{1}\left(\frac{1-q}{1-\qo}\right)\left(\frac{1}{2}+\frac{1}{4}\cdot\frac{1}{1-(1-q)}\right)\left(2\frac{1-q}{(1-\qo)^2}\right)dq\\
       &=\int_{\qo}^{1}\left(\frac{1-q}{1-\qo}\right)\left(\frac{1}{2}+\frac{1}{4}\cdot\sum_{n=0}^{\infty}{\left(1-q\right)^n}\right)\left(2\frac{1-q}{(1-\qo)^2}\right)dq\\
       &=\frac{1}{(1-\qo)^3}\int_{\qo}^{1}\left(1-q\right)\left(1+\frac{1}{2}\cdot\sum_{n=0}^{\infty}{\left(1-q\right)^n}\right)\left(1-q\right)dq\\
       &=\frac{1}{(1-\qo)^3}\int_{\qo}^{1}\left((1-q)^2+\frac{1}{2}\cdot\sum_{n=0}^{\infty}{\left(1-q\right)^{n+2}}\right)dq\\
       &=\frac{1}{(1-\qo)^3}\left(\frac{(1-\qo)^3}{3}+\frac{1}{2}\cdot\sum_{n=0}^{\infty}{\frac{(1-\qo)^{n+3}}{n+3}}\right)\\
       &=\frac{1}{3}+\frac{1}{2(1-\qo)^3}\cdot\sum_{n=3}^{\infty}{\frac{(1-\qo)^{n}}{n}}\\
       &=\frac{1}{3}+\frac{1}{2(1-\qo)^3}\cdot\left(\sum_{n=1}^{\infty}{\frac{(1-\qo)^{n}}{n}-\frac{(1-\qo)^2}{2}-(1-\qo)}\right)\\
       &=\frac{1}{3}-\frac{1}{4(1-\qo)}-\frac{1}{2(1-\qo)^2}-\frac{\log\qo}{2(1-\qo)^3}.\\
\end{align*}
Now consider the case in which $\qo<\nicefrac{2}{3}$.
\begin{align*}
E_2^R(r) &= \expect{q}{E_2^R(r\vert q)}\\
           &=\int_{\qo}^{1}{E_2^R(r\vert q)\cdot\left(2\frac{1-q}{(1-\qo)^2}\right)dq}\tag{by
\cref{l:min}
in \cref{s:minmax}}\\
           &\ge \int_{\qo}^{\nicefrac{2}{3}}{\!\!\!\!\!\!r(q)\left(1-\frac{1}{16}\cdot\frac{q}{1-q}\right)\cdot\left(2\frac{1-q}{(1-\qo)^2}\right)dq} + \int_{\nicefrac{2}{3}}^{1}{\!\!\!\!r(q)\left(\frac{1}{2}+\frac{1}{4}\cdot\frac{1}{q}\right)\cdot\left(2\frac{1-q}{(1-\qo)^2}\right)dq}\tag{by \cref{c:dec}}\\
           &\ge \int_{\qo}^{\nicefrac{2}{3}}{\left(\frac{1-q}{1-\qo}\right)\left(1-\frac{1}{16}\cdot\frac{q}{1-q}\right)\cdot\left(2\frac{1-q}{(1-\qo)^2}\right)dq} +\\*
            &\qquad\qquad\qquad\qquad\qquad\int_{\nicefrac{2}{3}}^{1}{\left(\frac{1-q}{1-\qo}\right)\left(\frac{1}{2}+ \frac{1}{4q}\right)\cdot\left(2\frac{1-q}{(1-\qo)^2}\right)dq}\tag{by \cref{r-concav}}\\
           &=\frac{2}{(1-\qo)^3}\int_{\qo}^{\nicefrac{2}{3}}{\left(1-q-\frac{q}{16}\right)\cdot\left(1-q\right)dq} + \frac{1}{(1-\qo)^3}\left(\frac{(\nicefrac{1}{3})^3}{3}+\frac{1}{2}\cdot\sum_{n=0}^{\infty}{\frac{(\nicefrac{1}{3})^{n+3}}{n+3}}\right)\tag{calculation of the second integral as in the case $\qo\ge\nicefrac{2}{3}$}\\
           &=\frac{1}{(1-\qo)^3}\left(\int_{\qo}^{\nicefrac{2}{3}}{\left(2-\frac{33}{8}q+\frac{17}{8}q^2\right)dq} + \frac{(\nicefrac{1}{3})^3}{3}+\frac{1}{2}\cdot\sum_{n=3}^{\infty}{\frac{(\nicefrac{1}{3})^n}{n}}\right)\\
           &=\frac{1}{(1-\qo)^3}\left(2\left(\frac{2}{3}-\qo\right)-\frac{33}{16}\left(\frac{4}{9}-\qo^2\right)+\frac{17}{24}\left(\frac{8}{27}-\qo^3\right) + \frac{1}{81}+\frac{1}{2}\cdot\sum_{n=3}^{\infty}{\frac{\left(\nicefrac{1}{3}\right)^n}{n}}\right)\\
           &=\frac{1}{(1-\qo)^3}\left(\frac{2}{3}(1-\qo)^3+\frac{2}{3}-\frac{33}{16}\cdot\frac{4}{9}+\frac{\qo^2}{16}+\frac{17}{24}\cdot\frac{8}{27}-\frac{\qo^3}{24} + \frac{1}{81}+\frac{1}{2}\cdot\sum_{n=3}^{\infty}{\frac{\left(\nicefrac{1}{3}\right)^n}{n}}\right)\\
           &=\frac{2}{3}+\frac{1}{(1-\qo)^3}\left(\frac{2}{3}-\frac{33}{16}\cdot\frac{4}{9}+\frac{17}{24}\cdot\frac{8}{27} + \frac{1}{81}+\frac{\qo^2}{16}-\frac{\qo^3}{24}+\frac{1}{2}\cdot\sum_{n=3}^{\infty}{\frac{\left(\nicefrac{1}{3}\right)^n}{n}}\right)\\
           &=\frac{2}{3}+\frac{1}{(1-\qo)^3}\left(\frac{-1}{36}+\frac{\qo^2}{16}-\frac{\qo^3}{24}+\frac{1}{2}\cdot\sum_{n=3}^{\infty}{\frac{\left(\nicefrac{1}{3}\right)^n}{n}}\right)\\
            &=\frac{2}{3}-\frac{1}{(1-\qo)^3}\left(\frac{1}{36}-\left(\frac{\qo}{4}\right)^2+\frac{1}{3}\left(\frac{\qo}{2}\right)^3-\frac{1}{2}\cdot\sum_{n=3}^\infty{\frac{\left(\nicefrac{1}{3}\right)^n}{n}}\right) \\
            &=\frac{2}{3}-\frac{1}{(1-\qo)^3}\left(\frac{1}{36}-\left(\frac{\qo}{4}\right)^2+\frac{1}{3}\left(\frac{\qo}{2}\right)^3+\frac{1}{2}\log\nicefrac{2}{3}+\frac{1}{2}\cdot\frac{\left(\nicefrac{1}{3}\right)^2}{2}+\frac{1}{2}\cdot\nicefrac{1}{3}\right) \\
            &=\frac{2}{3}-\frac{1}{(1-\qo)^3}\left(\frac{2}{9}-\left(\frac{\qo}{4}\right)^2+\frac{1}{3}\left(\frac{\qo}{2}\right)^3+\frac{\log\nicefrac{2}{3}}{2}\right).\tag*{\qed} \\
\end{align*}

\subsection{Calculation Omitted from the Proof of Lemma~\ref{l:left}}\label{s:calc-left}

To evaluate the first expectation in \cref{l-case2-ege}, we note that by \crefpart{l:max}{le} in \cref{s:minmax}, we have that $\expect{q}{\tfrac{3}{4}q ~\middle|~ q\leq2\gamma}=\frac{3}{4}\cdot\frac{2}{3}\cdot2\gamma=\gamma$. Evaluating the second expectation, we have
\begin{multline*}
\expect{q}{\left(\left(\frac{\gamma}{q}\right)^2-\frac{\gamma}{q}+1\right)\cdot q ~\middle|~  q > 2\gamma}=
\int_{2\gamma}^1 \left(\left(\frac{\gamma}{q}\right)^2-\frac{\gamma}{q}+1\right)\cdot q \cdot 2q \cdot dq = \\
= 2\int_{2\gamma}^1 \left(\gamma^2 - \gamma q + q^2\right)dq =
2\left(\gamma^2q - \gamma \frac{q^2}{2} + \frac{q^3}{3}\middle)\right|_{2\gamma}^1= \\
= 2\left(\gamma^2-\frac{\gamma}{2}+\frac{1}{3} - 2\gamma^3 + 2\gamma^3 - \frac{8}{3}\gamma^3\right) =
-\frac{16}{3}\gamma^3 + 2\gamma^2 - \gamma + \frac{2}{3}.
\end{multline*}
Plugging the expressions for these two expectations into \cref{l-case2-ege}, we obtain
\begin{align*}
E_2(e) &\ge 4\gamma^2\cdot\gamma + (1-4\gamma^2)\cdot \left(-\frac{16}{3}\gamma^3 + 2\gamma^2 - \gamma + \frac{2}{3}\right) \\
           &=\frac{64}{3}\gamma^5-8\gamma^4+\frac{8}{3}\gamma^3-\frac{2}{3}\gamma^2-\gamma+\frac{2}{3},
\end{align*}
as required.\qed

\subsection{Calculation Omitted from the Proof of Lemma~\ref{l:both}}\label{s:calc-both}

\begin{multline*}
\expect{(q_1,q_2)\sim U(B_1)}{e_2(x,y)} =
\expect{(x,y)\sim U([0,1]^2)}{G(x,y)} \ge
\frac{1+m}{2} - \frac{1}{2}\cdot\left(\frac{1+m}{2}\right)^2 = \\
\frac{2+\qo}{2(1+\qo)}-\frac{(2+\qo)^2}{8(1+\qo)^2} =
\frac{1}{2} + \frac{1}{2(1+\qo)}-\frac{(2+\qo)^2}{8(1+\qo)^2} = \\
\frac{1}{2} - \frac{1}{2}\left(\frac{(2+\qo)^2-4(1+\qo)}{(2(1+\qo))^2}\right) =
\frac{1}{2} - \frac{1}{2}\left(\frac{\qo^2}{(2(1+\qo))^2}\right) =
\frac{1}{2}\left(1-\left(\frac{\qo}{2(1+\qo)}\right)^2\right).\tag*{\qed}
\end{multline*}

\subsection{Calculation Omitted from the Proof of Theorem~\ref{more-than-half}}\label{s:calc-combine}

By \cref{e2-homogeneous} assume w.l.o.g.\ that $\OPT(F)=1$. We consider two cases based on the peak point $\qo$ of $r$.

If $\qo\ge\frac{2}{3}$, then by \cref{l:right,l:left,l:both}, we have that
\begin{align}
\ERM(F,2)&\ge (1-\qo)^2 \cdot \left(\frac{1}{3}-\frac{1}{4(1-\qo)}-\frac{1}{2(1-\qo)^2}-\frac{\log\qo}{2(1-\qo)^3}\right) + \nonumber\\*
&\qquad\qquad\qquad\qquad\qquad\qquad\qo^2\cdot0.528 + 2\qo(1-\qo)\cdot\frac{1}{2}\left(1-\left(\frac{\qo}{2\cdot(1+\qo)}\right)^2\right)= \nonumber\\
&=\left(\frac{1}{3}(1-\qo)^2-\frac{1}{4}(1-\qo)-\frac{1}{2}-\frac{\log\qo}{2(1-\qo)}\right) + \nonumber\\*
&\qquad\qquad\qquad\qquad\qquad\qquad\qo^2\cdot0.528 + \qo(1-\qo)\left(1-\left(\frac{\qo}{2\cdot(1+\qo)}\right)^2\right).\label{q-ge}
\end{align}
The above has a unique global minimum at $\qo\approx0.713832$, where it equals $\approx0.50922$.

If $\qo<\frac{2}{3}$, then by \cref{l:right,l:left,l:both}, we have that
\begin{align*}
\ERM(F,2)&\ge (1-\qo)^2 \cdot \left(\frac{2}{3}-\frac{1}{(1-\qo)^3}\left(\frac{2}{9}-\left(\frac{\qo}{4}\right)^2+\frac{1}{3}\left(\frac{\qo}{2}\right)^3+\frac{\log\nicefrac{2}{3}}{2}\right)\right) + \\*
&\qquad\qquad\qquad\qquad\qquad\qquad\qo^2\cdot0.528 + 2\qo(1-\qo)\cdot\frac{1}{2}\left(1-\left(\frac{\qo}{2\cdot(1+\qo)}\right)^2\right)= \\
&=\left(\frac{2}{3}(1-\qo)^2-\frac{1}{(1-\qo)}\left(\frac{2}{9}-\left(\frac{\qo}{4}\right)^2+\frac{1}{3}\left(\frac{\qo}{2}\right)^3+\frac{\log\nicefrac{2}{3}}{2}\right)\right) + \\*
&\qquad\qquad\qquad\qquad\qquad\qquad\qo^2\cdot0.528 + \qo(1-\qo)\left(1-\left(\frac{\qo}{2\cdot(1+\qo)}\right)^2\right).
\end{align*}
The derivative of the above is negative for $\qo\in[0,\nicefrac{2}{3}]$, and since at $\qo=\nicefrac{2}{3}$ it coincides with the lower-bound from \cref{q-ge}, we therefore have that $\ERM(F,2)>0.509$, regardless of the value of $\qo$.\qed

\section{Auxiliary Technical Results used in Sections~\ref{sec-r} and~\ref{sec-l}}\label{s:minmax}

\begin{lemma}[Properties of $\min\{q_1,q_2\}$]\label[lemma]{l:min}
Let $m\in[0,1]$, and let $(q_1,q_2)\sim U\bigl([m,1]^2\bigr)$. The random variable $\min\{q_1,q_2\}$ is a random attaining values in $[m,1]$,
whose
Cumulative Distribution Function is $\prob{}{\min\{q_1,q_2\}\leq q} = 1-\bigl(\frac{1-q}{1-m}\bigr)^2$,
and
whose
Probability Density Function is $\mu(q)=2\frac{1-q}{(1-m)^2}$.
\end{lemma}

\begin{lemma}[Properties of $\max\{q_1,q_2\}$]\label[lemma]{l:max}
Let $(q_1,q_2)\sim U\bigl([0,1]^2\bigr)$. The random variable $\max\{q_1,q_2\}$ is a random variable attaining values in $[0,1]$ with the following properties:
\begin{enumerate}
\item\label[part]{l:max-density} Its Cumulative Distribution Function is $\prob{}{\max\{q_1,q_2\}\leq q} = q^2$,
and its Probability Density Function is $\mu(q)=2q$.
\item\label[part]{l:max-le} For every $q\in [0,1]$,~ $\expect{q_1,q_2}{\max\{q_1,q_2\} \Big\vert \max\{q_1,q_2\}\leq q}= \frac{2}{3}q$.
\item\label[part]{l:max-ge} For every $q\in [0,1]$,~ $\expect{q_1,q_2}{\max\{q_1,q_2\} \Big\vert \max\{q_1,q_2\}\geq q}= \frac{2}{3}\cdot\frac{1-q^3}{1-q^2}$.
\end{enumerate}
\end{lemma}

\section{Proof of Lemma~\ref{e2-homogeneous}}\label{homogeneity-proof}

\begin{proof}[Proof of \cref{e2-homogeneous}]
For Part~3:
\[\OPT(\alpha\cdot r)=\max_{q\in[0,1]}\bigl(\alpha\cdot r(q)\bigr)=\alpha\cdot\max_{q\in[0,1]} r(q)=\alpha\cdot\OPT(r).\]
For Part~1, for every $q\in[0,1]$:
\[ v_{\alpha\cdot r}(q)=\frac{\alpha\cdot r(q)}{q}=\alpha\cdot\frac{r(q)}{q}=\alpha\cdot v_r(q). \]
For Part~2, let $q_1,q_2\in[0,1]$. By Part~1, $\max\bigl\{v_{\alpha\cdot r}(q_1),v_{\alpha\cdot r}(q_2)\bigr\}\geq 2\min\bigl\{v_{\alpha\cdot r}(q_1),v_{\alpha\cdot r}(q_2)\bigr\}$ if and only if $\max\bigl\{v_r(q_1),v_r(q_2)\bigr\}\geq 2\min\bigl\{v_r(q_1),v_r(q_2)\bigr\}$. We first consider the case where both of these conditions hold. In this case, by Part~1,
\[e_2^{\alpha\cdot r}(q_1,q_2)=\alpha\cdot r\left(\arg\smashoperator[l]{\max_{q\in\{q_1,q_2\}}}v_{\alpha\cdot r}(q)\right)=\alpha\cdot r\left(\arg\smashoperator[l]{\max_{q\in\{q_1,q_2\}}}v_r(q)\right)=\alpha\cdot e_2^r(q_1,a_2).\]
The case in which neither of these conditions holds is handled similarly, replacing $\arg\max$ with $\arg\min$.
\end{proof}

\end{document}